\documentclass[a4paper,UKenglish]{lipics}
 
\usepackage{etex}

\usepackage{microtype}

\usepackage[utf8]{inputenc}
\usepackage[T1]{fontenc}
\usepackage{amssymb,amsmath,amsthm}
\usepackage[toc,page]{appendix}
\usepackage{placeins}
\usepackage{graphicx}
\usepackage{algorithm,algorithmic}
\usepackage{tikz}
\usetikzlibrary{plotmarks}
\usetikzlibrary{automata} 
\usetikzlibrary[automata] 
\usetikzlibrary{decorations}
\usetikzlibrary{snakes}
\usetikzlibrary{decorations.pathmorphing}
\usepackage{pstricks,pst-node,pst-tree}
\usepackage{fancybox}
\usepackage{pgfkeys}

\usetikzlibrary{shapes.arrows,chains}

\usepackage[all]{xy}
\usepackage{bibentry}
\usepackage{graphics}
\usepackage{xspace,aeguill,ae}
\usepackage{verbatim,mathrsfs,enumerate,url}
\usepackage{nag}

\usepackage{wrapfig}

\newcommand{\rstrct}[2]{{#1}_{\upharpoonright_{#2}}} 

\newcommand{\IN}{\ensuremath{\mathbb{N}}\xspace}
\newcommand{\IR}{\ensuremath{\mathbb{R}}\xspace}
\newcommand{\IQ}{\ensuremath{\mathbb{Q}}\xspace}

\newcommand{\IH}{\mathbf{H}}
\newcommand{\ID}{\ensuremath{\mathbf{D}}\xspace}

\newcommand{\g}{\lambda}
\newcommand{\maxg}{\Lambda}
\newcommand{\barmaxg}{\overline{\Lambda}}

\newcommand{\WG}{\mathcal{G} = (\Pi, V, (V_i)_{i \in \Pi}, E, \bar \g)}
\newcommand{\SB}[2]{{#1}_{\upharpoonright{#2}}} 
\newcommand{\SBWG}{\rstrct{\mathcal{G}}{h} = (\Pi, V,(V_i)_{i \in \Pi}, E,  \rstrct{\bar \g}{h})} 
\newcommand{\out}[1]{\langle #1 \rangle}

\newcommand{\Target}{T}
\newcommand{\Sh}[2]{{#1}({#2})} 

\DeclareMathOperator{\Last}{Last}
\DeclareMathOperator{\First}{First}
\DeclareMathOperator{\Hist}{Hist}

\newcommand{\G}{\mathcal{G}}

\newcommand{\PotS}[3]{{\PotPlay}^{#1}_{#2}({#3})} 
\newcommand{\EraseS}[3]{{\ErasePlay}^{#1}_{#2}({#3})} 
\newcommand{\hg}[3]{h_{#1,#2,#3} \ccdot (g_{#1,#2,#3})^{\omega}}
\newcommand{\hgbis}[2]{h_{#1,#2} \ccdot (g_{#1,#2})^{\omega}}
\newcommand{\prefix}[2]{#1\widehat{~~}#2}

\newcommand{\PotPlay}{\mathbf{P}}
\newcommand{\ErasePlay}{\mathbf{E}}

\newcommand{\const}{c}
\newcommand{\Range}{C}

\newcommand{\devstep}[1]{$#1$-deviation step}

\newcommand{\ccdot}{\!\cdot\!}

\renewcommand{\phi}{\mathchar"11E}            

\title{Weak Subgame Perfect Equilibria and their Application to Quantitative Reachability\footnote{This work has been partly supported by European project Cassting (FP7-ICT-601148).}}
\titlerunning{Weak Subgame Perfect Equilibria in Quantitative Games} 

\author[1]{Thomas Brihaye}
\author[2]{V\'{e}ronique Bruy\`{e}re}
\author[2]{No\'{e}mie Meunier\thanks{Author supported by F.R.S.-FNRS fellowship.}}
\author[3]{Jean-Fran\c cois Raskin\thanks{Author supported by ERC Starting Grant (279499: inVEST).}}

\affil[1]{D\'{e}partement de math\'ematique, Universit\'{e} de Mons (UMONS)\\
  Mons, Belgium}
\affil[2]{D\'{e}partement d'informatique, Universit\'{e} de Mons (UMONS)\\
  Mons, Belgium}
\affil[3]{D\'{e}partement d'informatique, Universit\'{e} Libre de Bruxelles (U.L.B.)\\
  Brussels, Belgium}
\authorrunning{Th.\,Brihaye, V.\,Bruy\`{e}re, N.\,Meunier and J.-F.\,Raskin} 

\Copyright{Thomas Brihaye, V\'{e}ronique Bruy\`{e}re, No\'{e}mie Meunier and Jean-Fran\c cois Raskin}


\begin{document}

\maketitle

\begin{abstract} We study $n$-player turn-based games played on a finite directed graph. For each play, the players have to pay a cost that they want to minimize. Instead of the well-known notion of Nash equilibrium (NE), we focus on the notion of subgame perfect equilibrium (SPE), a refinement of NE well-suited in the framework of games played on graphs. We also study natural variants of SPE, named weak (resp. very weak) SPE, where players who deviate cannot use the full class of strategies but only a subclass with a finite number of (resp. a unique) deviation step(s). 

Our results are threefold. Firstly, we characterize in the form of a Folk theorem the set of all plays that are the outcome of a weak SPE. We also establish a weaker version of this theorem for SPEs. Secondly, for the class of quantitative reachability games, we prove the existence of a finite-memory SPE and provide an algorithm for computing it (only existence was known with no information regarding the memory). Moreover, we show that the existence of a constrained SPE, i.e. an SPE such that each player pays a cost less than a given constant, can be decided. The proofs rely on our Folk theorem for weak SPEs (which coincide with SPEs in the case of quantitative reachability games) and on the decidability of MSO logic on infinite words. Finally with similar techniques, we provide a second general class of games for which the existence of a (constrained) weak SPE is decidable.
\end{abstract}

 

\section{Introduction}

%

Two-player zero-sum infinite duration games played on graphs are a mathematical model used to formalize several important problems in computer science. Reactive system synthesis is one such important problem. In this context, see e.g.~\cite{PnueliR89}, the vertices and the edges of the graph represent the states and the transitions of the system; one player models the system to synthesize, and the other player models the (uncontrollable) environment of the system. In the classical setting, the objectives of the two players are opposite, i.e. the environment is \emph{adversarial}.   Modeling the environment as fully adversarial is usually a bold abstraction of reality and there are recent works that consider the more general setting of non zero-sum games which allow to take into account the different objectives of each player. In this latter setting the environment has its own objective which is most often \emph{not} the negation of the objective of the system. The concept of \emph{Nash equilibrium} (NE)~\cite{Nas50} is central to the study of non zero-sum games and can be applied to the general setting of $n$ player games.  A strategy profile is a NE if no player has an incentive to deviate unilaterally from his strategy, since he cannot strictly improve on the outcome of the strategy profile by changing his strategy only.


However in the context of sequential games (such as games played on graphs), it is well-known that NEs present a serious weakness: a NE allows for \emph{non-credible threats} that rational players should not carry out~\cite{Rubi91}.  As a consequence, for sequential games, the notion of NE has been strengthened into the notion of \emph{subgame perfect equilibrium} (SPE):  a strategy profile is an SPE if it is a NE in all the subgames of the original game. While the notion of SPE is rather well understood for finite state game graphs with $\omega$-regular objectives or for games in finite extensive form (finite game trees), less is known for game graphs with \emph{quantitative objectives} in which players encounter costs that they want to minimize, like in classical quantitative objectives such as mean-payoff, discounted sum, or quantitative reachability.

Several natural and important questions arise for such games: Can we decide the existence of an SPE, and more generally the \emph{constrained} existence of an SPE (i.e. an SPE in which each player encounters a cost less than some fixed value)? Can we compute such SPEs that use finite-memory strategies only? Whereas several work has studied what are the hypotheses to impose on games in a way to guarantee the existence of an SPE, the previous algorithmic questions are still widely open.
In this article, we provide progresses in the understanding of the notion of SPE. We first establish Folk theorems that characterize the possible outcomes of SPEs in quantitative games for SPEs and their variants. We then derive from this characterization interesting algorithms and information on the strategies for two important classes of quantitative games. Our contributions are detailed in the next paragraph.

%

\medskip\noindent{\bf Contributions~~}
First, we formalize a notion of \emph{deviation step} from a strategy profile that allows us to define two natural variants of NEs. While a NE must be resistant to the unilateral deviation of one player for any number of deviation steps, a \emph{weak} (resp. \emph{very weak}) NE must be resistant to the unilateral deviation of one player for any \emph{finite} number of (resp. a unique) deviation step(s).
Then we use those variants to define the corresponding notions of {\em weak} and {\em very weak} SPE. The latter notion is very close to the one-step deviation property \cite{Osborne94}.
Any very weak SPE is also a weak SPE, and there are games for which there exists a weak SPE but no SPE. Also, for games with upper-semicontinuous cost functions and for games played on finite game trees, the three notions are equivalent.


Second, we characterize in the form of a Folk theorem all the possible outcomes of weak SPEs. The characterization is obtained starting from all possible plays of the game and the application of a nonincreasing operator that removes plays that cannot be outcome of a weak SPE. We show that the limit of the nonincreasing chain of sets always exists and contains exactly all the possible outcomes of weak SPEs. Furthermore, we show how for each such outcome, we can associate a strategy profile that generates it and which is a weak SPE. Using a variant of the techniques developed for weak SPEs, we also obtain a weaker version of this Folk theorem for SPEs.


Additionally, to illustrate the potential of our Folk theorem, we show how it can be refined and used to answer open questions about two classes of quantitative games. 
The first class of games that we consider are {\em quantitative reachability games}, such that each player aims at reaching his own set of target states as soon as possible. As the cost functions in those games are continuous, our Folk theorem characterizes precisely the outcomes of SPEs and not only weak SPEs. In~\cite{BBD,Fudenberg83}, it has been shown that quantitative reachability games always have SPEs. The proof provided for this theorem is non constructive since it relies on topological arguments. Here, we strengthen this existential result by proving that there always exists, not only an SPE but, a {\em finite-memory} SPE. Furthermore, we provide an algorithm to construct such a finite memory SPE. This algorithm is based on a constructive version of our Folk Theorem for the class of quantitative reachability games: we show that the nonincreasing chain of sets of potential outcomes stabilizes after a finite number of steps and that each intermediate set is an $\omega$-regular set that can be effectively described using MSO sentences. 
The second class of games that we consider is the class of games with cost functions that are {\em prefix-independent}, whose range of values is {\em finite}, and for which each value has an {\em $\omega$-regular} pre-image. For this general class of games, with similar techniques as for quantitative reachability games, we show how to construct an effective representation of all possible outcomes compatible with a weak SPE, and consequently that the existence of a weak SPE is decidable.
In those two applications, we show that our construction also allow us to answer the question of existence of a constrained (weak) SPE, i.e. a (weak) SPE in which players pays a cost which is bounded by a given value.

\medskip\noindent{\bf Related work~~}
%
The concept of SPE has been first introduced and studied by the game theory community. The notion of SPE has been first introduced by Kuhn in finite extensive form games~\cite{kuhn53}. For such games, backward induction can be used to prove that there always exist an SPE. By inspecting the backward induction proof, it is not difficult to realize that the notion of very weak SPE and SPE are equivalent in this context.

SPEs for infinite trees defined as the unfolding of finite graphs with {\em qualitative}, i.e. win-lose, $\omega$-regular objectives, have been studied by Ummels in~\cite{Ummels06}: it is proved that such games always have an SPE, and that the existence of a constrained SPE is decidable. 


In~\cite{Kim}, Klimos et al. provide an effective representation of the outcomes of NEs in concurrent priced games by constructing a B\"uchi automaton accepting precisely the language of outcomes of all NEs satisfying a bound vector.
%
%
The existence of NEs in quantitative games played on graphs is studied in~\cite{TJS}; it is shown that for a large class of games, there always exists a finite-memory NE. This result is extended in~\cite{BMR14} for two-player games and secure equilibria (a refinement of NEs); additionally the constrained existence problem for secure equilibria is also shown decidable for a large range of cost functions. None of these references consider SPEs.


In~\cite{Flesch10}, the authors prove that for quantitative games with cost functions that are upper-semicontinus and with finite range, there always exits an SPE. This result also relies on a nonincreasing chain of sets of possible outcomes of SPEs. The main differences with our work is that we obtain a Folk theorem that \emph{characterize} all possible outcomes of weak SPEs with no restriction on the cost functions. Moreover we have shown that our Folk theorem can be made effective for two classes of quantitative games of interest. Effectiveness issues are not considered in~\cite{Flesch10}. 
Prior to this work, Mertens shows in~\cite{Mertens87} that if the cost functions are bounded and Borel measurable then there always exists an $\epsilon$-NE. In~\cite{Fudenberg83}, Fudenberg et al. show that if the cost functions are all continous, then there always exists an SPE. Those works were recently extended in~\cite{Roux14} by Le Roux and Pauly. 

\medskip\noindent{\bf Organization of the article~~}
In Section~\ref{section:variants}, we present the notions of quantitative game, classical NE and SPE, and their variants. In Section~\ref{section:FolkTheorem}, we propose and prove our Folk Theorems for weak SPEs and for SPEs. In Section~\ref{section:reach}, we provide an algorithm for computing a finite-memory SPE for quantitative reachability games, and a second algorithm to decide the constrained existence of an SPE for this class of games. In Section~\ref{section:prefixind}, we show that the existence of a (constrained) weak SPE is decidable for another class of games. A conclusion and future work are given in the last section.

\section{Preliminaries and Variants of Equilibria} \label{section:variants}

In this section, we recall the notions of quantitative game, Nash equilibrium, and subgame perfect equilibrium. We also introduce variants of Nash and subgame perfect equilibria, and compare them with the classical notions.

\subsection{Quantitative Games}
%
We consider multi-player turn-based non zero-sum quantitative games in which, for each infinite play, players pay a cost that they want to minimize.\footnote{Alternatively, players could receive a payoff that they want to maximize.} 

\begin{definition}
A \emph{quantitative game} is a tuple $\WG$ where: 

\begin{itemize}

  \item $\Pi $ is a finite set of players, 

  \item $V$ is a finite set of vertices, 

  \item $(V_i)_{i \in \Pi}$ is a partition of $V$ such that $V_i$ is the set of vertices controlled by player $i \in \Pi$, 

  \item $E \subseteq V \times V$ is a set of edges, such that\footnote{Each vertex has at least one outgoing edge.} for all $v \in V$, there exists $v' \in V$ with $(v,v') \in E$,

  \item $\bar \g = (\g_i)_{i \in \Pi}$ is a cost function such that $\g_i: V^\omega \to \IR \cup \{+\infty\}$ is player $i$ cost function.

\end{itemize}
\end{definition}

A \emph{play} of $\mathcal{G}$ is an infinite sequence $\rho = \rho_0 \rho_1 \ldots \in V^\omega$ such that 
$(\rho_i, \rho_{i + 1}) \in E$ for all $i \in \IN$. \emph{Histories} of $\mathcal{G}$ are finite 
sequences $h = h_0 \ldots h_n \in V^+$ defined in the same way. The \emph{length} $|h|$ of $h$ is the 
number $n$ of its edges. We denote by $\First(h)$ (resp. $\Last(h)$) the first vertex $h_0$ 
(resp. last vertex $h_n$) of $h$. Usually histories are non-empty, but in specific situations it will be useful 
to consider the empty history $\epsilon$. The set of all histories (ended by a vertex in $V_i$) is 
denoted by $\Hist$ (by $\Hist_i$). A \emph{prefix} (resp. \emph{suffix}) of a play $\rho$ is a finite 
sequence $\rho_0 \dots \rho_n$  (resp. infinite sequence $\rho_n \rho_{n+1} \ldots$) denoted by $\rho_{\leq n}$ 
or $\rho_{< n +1}$ (resp. $\rho_{\geq n}$). We use notation $h < \rho$ when a history $h$ is prefix of a play $\rho$. 
Given two distinct plays $\rho$ and $\rho'$, their longest common prefix is denoted by $\prefix{\rho}{\rho'}$.  

When an initial vertex $v_0 \in V$ is fixed, we call $(\G, v_0)$ an \emph{initialized} quantitative game. A play (resp. a history) of $(\G, v_0)$ is a play (resp. history) of $\G$ starting in $v_0$. The set of histories $h \in \Hist$ (resp. $h \in \Hist_i$) with $\First(h) = v_0$ is denoted by $\Hist(v_0)$ (resp. $\Hist_i(v_0)$). In the figures of this article, we will often \emph{unravel} the graph of the game $(\G, v_0)$ from the initial vertex $v_0$, which ends up in an infinite tree.



Given a play $\rho\in V^\omega$, its \emph{cost} is given by $\bar \g(\rho) = (\g_i(\rho))_{i \in \Pi}$. In this article, we are particularly interested in quantitative reachability games in which $\g_i(\rho)$ is equal to the number of edges to reach a given set of vertices.

\begin{definition} \label{def:reach}
A \emph{quantitative reachability game} is a quantitative game $\G$ such that the cost function $\bar \g : V^\omega \to (\IN \cup \{+\infty\})^{\Pi}$ is defined as follows. Each player~$i$ has a \emph{target set} $\Target_i \subseteq V$, and for each play $\rho = \rho_0\rho_1\ldots$ of $\G$, the cost $\g_i(\rho)$ is the least index $n$ such that $\rho_n \in \Target_i$ if it exists, and $+ \infty$ otherwise. 
\end{definition}

\noindent
Notice that the cost function $\bar \g$ of a quantitative game is often defined from $|\Pi|$-uples of weights labeling the edges of the game. For instance, in inf games, $\g_i(\rho)$ is equal to the infimum of player $i$ weights seen along~$\rho$. Some other classical examples are liminf, limsup, mean-payoff, and discounted sum games~\cite{LaurentDoyen}. In case of quantitative reachability on graphs with weighted edges, the cost $\g_i(\rho)$ for player $i$ is replaced by the sum of the weights seen along $\rho$ until his target set is reached. We do not consider this extension here. Notice that when weights are positive integers, replacing each edge with cost $c$ by a path of length $c$ composed of $c$ new edges allows to recover Definition~\ref{def:reach}.

%
%

Let us recall the notions of prefix-independent, continuous, and lower- (resp. upper-) semicontinuous cost functions. Since $V$ is endowed with the discrete topology, and thus $V^\omega$ with the product topology, a sequence of plays $(\rho_n)_{n \in \IN}$ converges to a play $\rho = \lim_{n \rightarrow \infty} \rho_n$ if every prefix of $\rho$ is prefix of all $\rho_n$ except, possibly, of finitely many of them.

\begin{definition} \label{def:kindCost}
Let $\g_i$ be a player~$i$ cost function. Then 
\begin{itemize}
\item $\g_i$ is \emph{prefix-independent} if $\g_i(h \rho) = \g_i(\rho)$ for any history $h$ and play $\rho$. 
\item $\g_i$ is \emph{continuous} if whenever $\lim_{n \rightarrow \infty} \rho_n =  \rho$, then $\lim_{n \rightarrow \infty} \g_i(\rho_n) =  \g_i(\rho)$. 
\item $\g_i$ \emph{upper-semicontinuous} (resp. \emph{lower-semicontinuous})  if whenever $\lim_{n \rightarrow \infty} \rho_n =  \rho$, then $\limsup_{n \rightarrow \infty} \g_i(\rho_n) \leq \g_i(\rho)$ (resp. $\liminf_{n \rightarrow \infty} \g_i(\rho_n) \geq \g_i(\rho)$).
\end{itemize}
\end{definition}

\noindent 
For instance, the cost functions used in liminf and mean-payoff games are prefix-independent, contrarily to the case of inf games.
Clearly, if $\g_i$ is continuous, then it is upper- and lower-semicontinuous. For instance, the cost functions of liminf and mean-payoff games are neither upper-semicontinuous nor lower-semicontinuous, whereas cost functions of discounted sum games are continuous. The cost functions $\g_i$, $i \in \Pi$, used in quantitative reachability games can be transformed into continuous ones as follows~\cite{BBD}: $\g'_i(\rho) = 1 - \frac{1}{\g_i(\rho) + 1}$ if $\g_i(\rho) < +\infty$, and $\g'_i(\rho) = 1$ otherwise.

\subsection{Strategies and Deviations}

A \emph{strategy} $\sigma$ for player $i \in \Pi$ is a function $\sigma: \Hist_i \to V$ assigning to each history\footnote{In this article we often write a history in the form $hv$ with $v \in V$ to emphasize that $v$ is the last vertex of this history.}  $hv \in \Hist_i$ a vertex $v' = \sigma(hv)$ such that $(v, v') \in E$.
In an initialized game $(\mathcal{G}, v_0)$, $\sigma$ is restricted to histories starting with $v_0$.
A player $i$ strategy $\sigma$ is \emph{positional} if it only depends on the last vertex of the history, i.e. $\sigma(hv) = \sigma(v)$ for all $hv \in \Hist_i$. It is a \emph{finite-memory} strategy if it needs only finite memory of the history (recorded by a finite strategy automaton, also called a Moore machine). A play $\rho$ is \emph{consistent} with a player~$i$ strategy $\sigma$ if $\rho_{k+1} = \sigma(\rho_{\leq k})$ for all $k$ such that $\rho_k \in V_i$. A \emph{strategy profile} of $\mathcal{G}$ is a tuple $\bar\sigma = (\sigma_i)_{i \in \Pi}$ of strategies, where each $\sigma_i$ is a player~$i$ strategy. It is called \emph{positional} (resp. \emph{finite-memory}) if all $\sigma_i$, $i \in \Pi$, are positional (resp. finite-memory). Given an initial vertex $v_0$, such a strategy profile determines a unique play of $(\mathcal{G}, v_0)$ that is consistent with all the strategies. This play is called the \emph{outcome} of $\bar\sigma$ and is denoted by $\out{\bar\sigma}_{v_0}$. 
 
Given $\sigma_i$ a player $i$ strategy, we say that player $i$ \emph{deviates} from $\sigma_i$ if he does not stick to $\sigma_i$ and prefers to use another strategy $\sigma'_i$. Let $\bar \sigma$ be a strategy profile. When all players stick to their strategy $\sigma_i$ except player $i$ that shifts to $\sigma'_i$, we denote by $(\sigma'_i, \sigma_{-i})$ the derived strategy profile, and by $\out{\sigma'_i, \sigma_{-i}}_{v_0}$ its outcome in $(\mathcal{G}, v_0)$. In the next definition, we introduce the notion of deviation step of a strategy $\sigma'_i$ from a given strategy profile $\bar\sigma$.

\begin{definition} \label{def:devstep}
Let $(\G, v_0)$ be an initialized game, $\bar \sigma$ be a strategy profile, and $\sigma'_i$ be a player $i$ strategy. We say that $\sigma'_i$ has a \emph{\devstep{hv}} from $\bar \sigma$ for some history $hv \in \Hist_i(v_0)$ with $v \in V_i$, if 
$$hv <  \out{\sigma'_i,\sigma_{-i}}_{v_0} \mbox{ and } \sigma_i(hv) \ne \sigma'_i(hv).$$
\end{definition}

\noindent
Notice that the previous definition requires that $hv$ is a prefix of the outcome $\out{\sigma'_i,\sigma_{-i}}_{v_0}$; 
it says nothing about $\sigma'_i$ outside of this outcome. A strategy $\sigma'_i$ can have a finite or an infinite number 
of deviation steps in the sense of Definition~\ref{def:devstep}. A strategy with three deviation steps is depicted in 
Figure~\ref{fig:deviationsteps} (left) such that each \devstep{h_kv_k}  from $\bar\sigma$, $1\leq k \leq3$, 
is highlighted with a dashed edge. We will come back to this figure later on. 

\begin{figure}[h!] 
\begin{minipage}[c]{.46\linewidth}
\centering
\begin{tikzpicture}[initial text=,auto, node distance=1cm, shorten >=1pt, scale=0.1] 

\node[state, scale=0.5]                           (0)                                  {$v_0$};
\node                                            (0bbis)  [right=0.2cm of 0]           {};
\node[state, scale=0.5]                           (1)     [below=1cm of 0bbis]         {$v_1$};
\node[scale=0.6]                                  (1bis)  [below=0.6cm of 1]           {};
\node[state, scale=0.5]                           (3)     [below right=0.7cm of 1]     {$v'_1$};
\node                                            (1bbis)  [right=0.2cm of 3]           {};
\node[state, scale=0.5]                           (4)     [below=1cm of 1bbis]         {$v_2$};
\node[scale=0.6]                                  (2bis)  [below=0.6cm of 4]           {};
\node[state, scale=0.5]                           (5)     [below right=0.7cm of 4]     {$v'_2$};
\node                                            (2bbis)  [right=0.2cm of 5]           {};
\node[state, scale=0.5]                           (6)     [below=1cm of 2bbis]         {$v_3$};
\node[scale=0.6]                                  (3bis)  [below=0.6cm of 6]           {};
\node[state, scale=0.5]                           (7)     [below right=0.7cm of 6]     {$v'_3$};
\node                                            (7bis)   [right=0.2cm of 7]           {};
\node[scale=0.6]                                  (8)     [below=2cm of 7bis]          {$\out{\sigma'_i,\sigma_{-i}}_{v_0}$};

\node[scale=0.5] (fictif1) at (30,3) {};
\node[scale=0.5] (fictif2) at (30,-35) {};

\node[scale=0.5] (fictif3) at (40,3) {};
\node[scale=0.5] (fictif4) at (40,-56) {};

\path[-]  (0)   edge      [out=-10,in=100]         node[scale=0.6]          {$h_1$}     (1)

          (1)   edge      [->, dashed]             node[left, scale=0.6]    {}        (3)
				
				edge      [->]                     node[left, scale=0.6]    {}        (1bis)

          (4)   edge      [->, dashed]             node[left, scale=0.6]    {}        (5)

                edge      [->]                     node[left, scale=0.6]    {}        (2bis)

          (3)   edge      [out=-10,in=100]         node[right, scale=0.6]   {}        (4)

          (6)   edge      [->, dashed]             node[left, scale=0.6]    {}        (7)

                edge      [->]                     node[left, scale=0.6]    {}        (3bis)

          (5)   edge      [out=-10,in=100]         node[right, scale=0.6]   {}        (6)

          (7)   edge      [out=-35,in=95]          node[right, scale=0.6]   {}        (8)

          (fictif1)  edge [<->]                    node[scale=0.6]          {$h_2$}   (fictif2)

          (fictif3)  edge [<->]                    node[scale=0.6]          {$h_3$}   (fictif4);

\end{tikzpicture}
\end{minipage} \hfill
\begin{minipage}[c]{.46\linewidth}
	\begin{tikzpicture}[initial text=,auto, node distance=1cm, shorten >=1pt, scale=0.1] 

	\node[state, scale=0.5]                           (0)                                  {$v_0$};
	\node                                            (0bbis)  [right=0.2cm of 0]           {};
	\node[state, scale=0.5]                           (1)     [below=1cm of 0bbis]         {$v_1$};
	\node[scale=0.6]                                  (1bis)  [below=2cm of 1]             {$h_1 \ccdot \out{\SB{\bar\sigma}{h_1}}_{v_1 = \out{\bar\sigma}_{v_0}}$};
	\node[state, scale=0.5]                           (3)     [below right=0.7cm of 1]     {$v'_1$};
	\node                                            (1bbis)  [right=0.2cm of 3]           {};
	\node[state, scale=0.5]                           (4)     [below=1cm of 1bbis]         {$v_2$};
	\node[scale=0.6]                                  (2bis)  [below=2cm of 4]             {$h_2 \ccdot \out{\SB{\bar\sigma}{h_2}}_{v_2} = h_1v_1 \ccdot \out{\SB{\bar\sigma}{h_1v_1}}_{v'_1}$};
	\node[state, scale=0.5]                           (5)     [below right=0.7cm of 4]     {$v'_2$};
	\node                                            (2bbis)  [right=0.2cm of 5]           {};
	\node[state, scale=0.5]                           (6)     [below=1cm of 2bbis]         {$v_3$};
	\node[scale=0.6]                                  (3bis)  [below=2cm of 6]             {$h_3 \ccdot \out{\SB{\bar\sigma}{h_3}}_{v_3} = h_2v_2 \ccdot \out{\SB{\bar\sigma}{h_2v_2}}_{v'_2}$};
	\node[state, scale=0.5]                           (7)     [below right=0.7cm of 6]     {$v'_3$};
	\node                                            (7bis)   [right=0.2cm of 7]           {};
	\node[scale=0.6]                                  (8)     [below=2cm of 7bis]         {$\out{\sigma'_i,\sigma_{-i}}_{v_0} = h_3v_3 \ccdot \out{\SB{\bar\sigma}{h_3v_3}}_{v'_3}$};

	\node             at (30,3)                       (10)          {};
	\node             at (30,-35)                     (11)          {};
	\node             at (40,3)                       (12)          {};
	\node             at (40,-56)                     (13)          {};

	\path[-]  (0)   edge      [out=-10,in=100]         node[scale=0.6]          {$h_1$}     (1)

	          (1)   edge      [->, dashed]             node[left, scale=0.6]    {}        (3)

					edge                               node[left, scale=0.6]    {}        (1bis)

	          (4)   edge      [->, dashed]             node[left, scale=0.6]    {}        (5)

	                edge                               node[left, scale=0.6]    {}        (2bis)

	          (3)   edge      [out=-10,in=100]         node[right, scale=0.6]   {}        (4)

	          (6)   edge      [->, dashed]             node[left, scale=0.6]    {}        (7)

	                edge                               node[left, scale=0.6]    {}        (3bis)

	          (5)   edge      [out=-10,in=100]         node[right, scale=0.6]   {}        (6)

	          (7)   edge      [out=-35,in=95]          node[right, scale=0.6]   {}        (8)
	
		      (10)   edge      [<->]                   node[right, scale=0.6]   {$h_2$}        (11)
		
		      (12)   edge      [<->]                   node[right, scale=0.6]   {$h_3$}        (13);

	\end{tikzpicture}
\end{minipage}
\caption{A strategy $\sigma'_i$ with a finite number of deviation steps.}\label{fig:deviationsteps}
\end{figure}
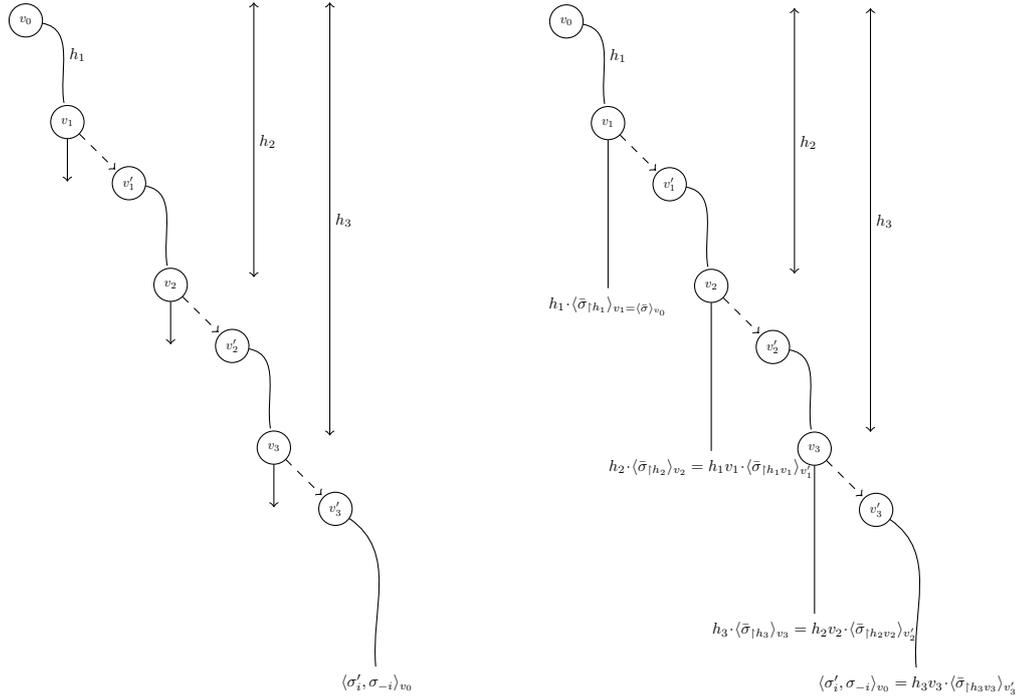

In light of Definition~\ref{def:devstep}, we introduce the following classes of strategies.

\begin{definition}
Let $(\G, v_0)$ be an initialized game, and $\bar \sigma$ be a strategy profile. 
\begin{itemize}
\item A strategy $\sigma'_i$ is \emph{finitely deviating} from $\bar \sigma$ if it has a finite number of deviation steps from $\bar \sigma$.
\item It is \emph{one-shot deviating} from $\bar \sigma$ if it has a \devstep{v_0} from $\bar \sigma$, and no other deviation step.
\end{itemize}
\end{definition}

\noindent
In other words, a strategy $\sigma'_i$ is finitely deviating from $\bar \sigma$ if there exists a history $hv < \out{\sigma'_i,\sigma_{-i}}_{v_0}$ such that for all $h'v'$, $hv \leq h'v' < \out{\sigma'_i,\sigma_{-i}}_{v_0}$, we have $\sigma'_i(h'v') = \sigma_i(h'v')$ ($\sigma'_i$ acts as $\sigma_i$ from $hv$ along $\out{\sigma'_i,\sigma_{-i}}_{v_0}$). The strategy $\sigma'_i$ is one-shot deviating from $\bar \sigma$ if it differs from $\sigma_i$ at the initial vertex $v_0$, and after $v_0$ acts as $\sigma_i$ along $\out{\sigma'_i,\sigma_{-i}}_{v_0}$. As for Definition~\ref{def:devstep}, the previous definition says nothing about $\sigma'_i$ outside of $\out{\sigma'_i,\sigma_{-i}}_{v_0}$. Clearly any one-shot deviating strategy is finitely deviating. The strategy of Figure~\ref{fig:deviationsteps} is finitely deviating but not one-shot deviating.

\subsection{Nash and Subgame Perfect Equilibria, and Variants}

In this paper, we focus on subgame perfect equilibria and their variants. Let us first recall the classical notion of Nash equilibrium. 
A strategy profile $\bar\sigma$ in an initialized game is a Nash equilibrium if no player has an incentive to deviate 
unilaterally from his strategy, since he cannot strictly decrease his cost when using any other strategy. 

\begin{definition}
Given an initialized game $(\G, v_0)$, a strategy profile $\bar \sigma = (\sigma_i)_{i \in \Pi}$ of $(\G,v_0)$ is a \emph{Nash equilibrium (NE)} if for all players $i \in \Pi$, for all player $i$ strategies $\sigma'_i$, we have 
$\g_i(\out{\sigma'_i, \sigma_{-i}}_{v_0}) \geq \g_i(\out{\bar \sigma}_{v_0})$.
\end{definition}

\noindent
We say that a player~$i$ strategy $\sigma'_i$ is a \emph{profitable deviation} for $i$ w.r.t. $\bar\sigma$ if $\g_i(\out{\sigma'_i, \sigma_{-i}}_{v_0}) < \g_i(\out{\bar \sigma}_{v_0})$. Therefore $\bar\sigma$ is a NE if no player has a profitable deviation w.r.t. $\bar\sigma$. 

Let us propose the next variants of NE.

\begin{definition}
Let $(\G, v_0)$ be an initialized game. A strategy profile $\bar \sigma$ is a \emph{weak NE} (resp. \emph{very weak NE}) in $(\G,v_0)$ if, for each player $i \in \Pi$, 
for each finitely deviating (resp. one-shot deviating) strategy $\sigma'_i$ of player~$i$, we have
$\g_i(\out{\sigma'_i, \sigma_{-i}}_{v_0}) \geq \g_i(\out{\bar \sigma}_{v_0})$.
\end{definition}

\begin{example}
Consider the two-player quantitative game depicted in Figure~\ref{TwoPlayerQuantitativeGame}. 
Circle (resp. square) vertices are player $1$ (resp. player $2$) vertices. The edges are labeled by couples of weights such that weights $(0, 0)$ are not specified. For each player~$i$, the cost $\g_i(\rho)$ of a play $\rho$ is the weight of its ending loop. In this simple game, each player~$i$ have two positional strategies that are respectively denoted by $\sigma_i$ and $\sigma'_i$ (see Figure \ref{TwoPlayerQuantitativeGame}). 

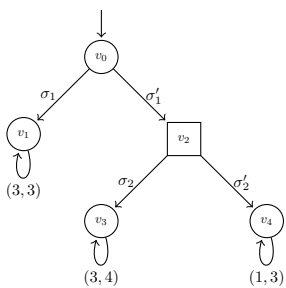
\begin{figure}[h!]
\begin{center}
\begin{tikzpicture}[initial text=,auto, node distance=1cm, shorten >=1pt] 

\node[state, scale=0.5]               (0)                          {$v_0$};
\node[state, scale=0.5]               (1)    [below left=of 0]     {$v_1$};
\node[state, rectangle, scale=0.5]    (2)    [below right=of 0]    {$v_2$};
\node[state, scale=0.5]               (3)    [below left=of 2]     {$v_3$};
\node[state, scale=0.5]               (4)    [below right=of 2]    {$v_4$};

\node (fictif) [above=4 mm of 0]  {};

\path[->] (0)  edge                      node[left, scale=0.6]           {$\sigma_1$}          (1)
               edge                      node[right, scale=0.6]          {$\sigma'_1$}         (2)

          (1)  edge  [loop below]        node[midway, scale=0.6]         {$(3, 3)$}            ()

          (2)  edge                      node[left, scale=0.6]           {$\sigma_2$}          (3)
               edge                      node[right, scale=0.6]          {$\sigma'_2$}         (4)

          (3)  edge  [loop below]        node[midway, scale=0.6, black]  {$(3, 4)$}            ()

          (4)  edge  [loop below]        node[midway, scale=0.6, black]  {$(1, 3)$}            ()

          (fictif)   edge (0);

\end{tikzpicture}
\caption{A simple two-player quantitative game\label{TwoPlayerQuantitativeGame}}
\end{center}
\end{figure}

The strategy profile $(\sigma_1, \sigma'_2)$ is not a NE since $\sigma'_1$ is a profitable deviation for player~$1$ w.r.t. $(\sigma_1, \sigma'_2)$ (player~$1$ pays cost $1$ instead of cost $3$). This strategy profile is neither a weak NE nor a very weak NE because in this simple game, player~$1$ can only deviate from $\sigma_1$ by using the one-shot deviating strategy $\sigma'_1$.
On the contrary, the strategy profile $(\sigma_1, \sigma_2)$ is a NE with outcome $v_0v_1^\omega$ of cost $(3,3)$. It is also a weak NE and a very weak NE.  
\end{example}

By definition, any NE is a weak NE, and any weak NE is a very weak NE. The contrary is false: in the previous example, $(\sigma'_1,\sigma_2)$ is a very weak NE, but not a weak NE. We will see later an example of game with a weak NE that is not an NE (see Example~\ref{ex:contrex}).

The notion of subgame perfect equilibrium is a refinment of NE. In order to define it, we need to introduce the following notions. Given a quantitative game $\WG$ and a history $h$ of $\G$, we denote by $\SB{\G}{h}$ the game $\SBWG$ where 
$$\SB{\bar\g}{h}(\rho) = \bar\g(h\rho)$$ 
for any play of $\SB{\G}{h}$\footnote{In this article, we will always use notation $\bar\g(h\rho)$ instead of $\SB{\bar\g}{h}(\rho)$.}, and we say that $\SB{\G}{h}$ is a \emph{subgame} of $\G$. Given an initialized game $(\G, v_0)$, and a history $hv \in \Hist(v_0)$, the initialized game $(\SB{\G}{h}, v)$ is called the subgame of $(\G, v_0)$ with history $hv$. Notice that $(\G, v_0)$ can be seen as a subgame of itself with history $hv_0$ such that $h = \epsilon$. Given a player $i$ strategy $\sigma$ in $(\G,v_0)$, we define 
the strategy $\SB{\sigma}{h}$ in $(\SB{\G}{h}, v)$ as $\SB{\sigma}{h}(h') = \sigma(hh')$ for all histories $h' \in \Hist_i(v)$. Given a strategy profile $\bar \sigma = (\sigma_i)_{i \in \Pi}$, we use notation $\SB{\bar \sigma}{h}$ for $(\SB{\sigma_i}{h})_{i \in \Pi}$, and $\out{\SB{\bar \sigma}{h}}_{v}$ is its outcome in the subgame $(\SB{\G}{h}, v)$.

We can now recall the classical notion of subgame perfect equilibrium: it is a strategy profile in an initialized game that induces a NE in each of its subgames. In particular, a subgame perfect equilibrium is a NE.

\begin{definition}
Given an initialized game $(\G, v_0)$, a strategy profile $\bar \sigma$ of $(\G,v_0)$ is a \emph{subgame perfect equilibrium (SPE)} if $\SB{\bar \sigma}{h}$ is a NE in $(\SB{\G}{h}, v)$, for every history $hv \in \Hist(v_0)$. 
\end{definition}

As for NE, we propose the next variants of SPE.

\begin{definition} \label{def:weakSPE}
Let $(\G, v_0)$ be an initialized game.  A strategy profile $\bar \sigma$ is a \emph{weak SPE} (resp. \emph{very weak SPE}) if $\SB{\bar \sigma}{h}$ is a weak NE (resp. very weak NE) in $(\SB{\G}{h}, v)$, for all histories $hv \in \Hist(v_0)$.
\end{definition}

\begin{example}
We come back to the game depicted in Figure~\ref{TwoPlayerQuantitativeGame}. We have seen before that the strategy profile $(\sigma_1, \sigma_2)$ is a NE. However it is not an SPE. Indeed consider the subgame $(\SB{\G}{v_0}, v_2)$ of $(\G,v_0)$ with history $v_0v_2$. In this subgame, $\sigma'_2$ is a profitable deviation for player~$2$. One can easily verify that the strategy profile $(\sigma'_1, \sigma'_2)$ is an SPE, as well as a weak SPE and a very weak SPE, due to the simple form of the game.
\end{example}

The previous example is too simple to show the differences between classical SPEs and their variants. The next example presents a game with a (very) weak SPE but no SPE. 

\begin{example} \label{ex:contrex}
Consider the initialized two-player game $(\G, v_0)$ in Figure~\ref{fig:gameNoSPE}.
The edges are labeled by couples of weights, and for each player~$i$ the cost $\g_i(\rho)$ of a play $\rho$ 
is the unique weight seen in its ending cycle. With this definition, $\g_i(\rho)$ can also be seen as either the mean-payoff, or the liminf, or the limsup, of the weights of $\rho$. 
It is known that this game has no SPE~\cite{SV03}. 

Let us show that the positional 
strategy profile $\bar \sigma$ depicted with thick edges is a very weak SPE. Due to the simple form of the game, 
only two cases are to be treated. Consider the subgame $(\SB{\G}{h}, v_0)$ with $h \in (v_0v_1)^\ast$, 
and the one-shot deviating strategy $\sigma'_1$ of player~$1$ such that $\sigma'_1(v_0) = v_2$. 
Then $\out{\SB{\bar\sigma}{h}}_{v_0} = v_0v_1v_3^\omega$ and $\out{\sigma'_1,\SB{\sigma_2}{h}}_{v_0} = v_0v_2^\omega$, 
showing that $\sigma'_1$ is not a profitable deviation for player~$1$. One also checks that in the 
subgame $(\SB{\G}{h}, v_1)$ with $h \in (v_0v_1)^\ast v_0$, the one-shot deviating strategy $\sigma'_2$ 
of player~$2$ such that $\sigma'_2(v_1) = v_0$ is not profitable for him.

Similarly, one can prove that $\bar \sigma$ is a weak SPE (see also Proposition~\ref{prop:weak-veryweak} hereafter). Notice that $\bar \sigma$ is not an SPE. Indeed the strategy $\sigma'_2$ such that $\sigma'_2(hv_1) = v_0$ for all $h$, is a profitable deviation for player 2 in $(\G, v_0)$. This strategy is (of course) not finitely deviating. Finally notice that $\bar \sigma$ is a weak NE that is not an NE.

\begin{figure}[ht!]
\begin{center}
\begin{tikzpicture}[initial text=,auto, node distance=2cm, shorten >=1pt] 

\node[state, scale=0.5]              (1)                     {$v_0$};
\node[state, rectangle, scale=0.5]   (2)    [right=of 1]     {$v_1$};
\node[state, scale=0.5]              (3)    [left=of 1]      {$v_2$};
\node[state, scale=0.5]              (4)    [right=of 2]     {$v_3$};

\node (fictif) [above left=4 mm of 1]  {};

\path[->] (1) edge [bend right=25, thick, black]          node[below, scale=0.6, black]    {$(2, 0)$}   (2)
              edge                                node[above, scale=0.6]           {$(0, 0)$}   (3)

          (2) edge  [bend left=-25]               node[above, scale=0.6]           {$(2, 0)$}   (1)
              edge  [thick, black]                        node[above, scale=0.6, black]    {$(0,0)$}   (4)

          (3) edge  [loop above, thick, black]            node[midway, scale=0.6, black]   {$(1, 2)$}   ()

          (4) edge  [loop above, thick, black]            node[midway, scale=0.6, black]   {$(0, 1)$}   ()

          (fictif)   edge (1);

\end{tikzpicture}
\end{center}
\caption{A two-player game with a (very) weak SPE and no SPE. For each player, the cost of a play 
is his unique weight seen in the ending cycle.}
\label{fig:gameNoSPE}
\end{figure}
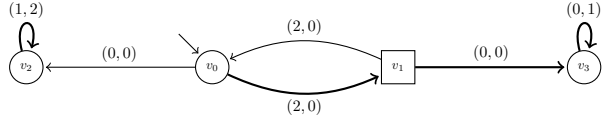
\end{example}

From Definition~\ref{def:weakSPE}, any SPE is a weak SPE, and any weak SPE is a very weak SPE. The next proposition states that weak SPE and very weak SPE are equivalent notions, but this is no longer true for SPE and weak SPE as shown previously by Example~\ref{ex:contrex}.

\begin{proposition} \label{prop:weak-veryweak}
\begin{itemize}
\item Let $(\G, v_0)$ be an initialized game, and $\bar \sigma$ be a strategy profile. Then $\bar \sigma$ is a weak SPE iff $\bar \sigma$ is a very weak SPE.
\item There exists an initialized game $(\G, v_0)$ with a weak SPE but no SPE.
\end{itemize}
\end{proposition}

Before proving this proposition, we would like to come back to the definition of deviation step (Definition~\ref{def:devstep}), and explain it now with the concept of subgame. Given an initialized game $(\G,v_0)$, a strategy profile $\bar \sigma$, and a player $i$ strategy $\sigma'_i$, we recall that $\sigma'_i$ has a \devstep{hv} from $\bar \sigma$ for some $hv \in \Hist_i(v_0)$, if $hv <  \out{\sigma'_i,\sigma_{-i}}_{v_0}$ and $\sigma_i(hv) \ne \sigma'_i(hv)$. Equivalently, $\sigma'_i$ has a \devstep{hv} from $\bar \sigma$ iff 
$$\prefix{h \ccdot \out{\SB{\bar\sigma}{h}}_{v}}{\out{\sigma'_i,\sigma_{-i}}_{v_0}} = hv.$$
This alternative vision of deviation step is depicted in Figure~\ref{fig:deviationsteps} (right) for a strategy with three deviation steps. For instance, for history $h_2v_2$, we have $h_2v_2 <  \out{\sigma'_i,\sigma_{-i}}_{v_0}$ and $\sigma_i(h_2v_2) \ne \sigma'_i(h_2v_2)$, or equivalently $\prefix{h_2 \ccdot \out{\SB{\bar\sigma}{h_2}}_{v_2}}{\out{\sigma'_i,\sigma_{-i}}_{v_0}} = h_2v_2$. Notice that there is no intermediate deviation step between the \devstep{h_1v_1} and the \devstep{h_2v_2} since $h_2 \ccdot \out{\SB{\bar\sigma}{h_2}}_{v_2} = h_1v_1 \ccdot \out{\SB{\bar\sigma}{h_1v_1}}_{v'_1}$ as indicated in the figure. 
Similarly the \devstep{h_1v_1} is the first one because $h_1 \ccdot \out{\SB{\bar\sigma}{h_1}}_{v_1} = \out{\bar\sigma}_{v_0}$, and
the \devstep{h_3v_3} is the third one because $h_3 \ccdot \out{\SB{\bar\sigma}{h_3}}_{v_3} = h_2v_2 \ccdot \out{\SB{\bar\sigma}{h_2v_2}}_{v'_2}$. The latter deviation step is the last one because $\out{\sigma'_i,\sigma_{-i}}_{v_0} = h_3v_3 \ccdot \out{\SB{\bar\sigma}{h_3v_3}}_{v'_3}$.

\begin{proof}[Proof of Proposition~\ref{prop:weak-veryweak}]  
This proof is based on arguments from the one-step deviation property used to prove Kuhn's theorem \cite{kuhn53}.
Let $\bar \sigma$ be a very weak SPE, and let us prove that it is a weak SPE. As a contradiction, assume that there exists a subgame $(\SB{\G}{h}, v)$ such that the strategy profile $\SB{\bar\sigma}{h}$ is not a weak NE. This means that there exists a player $i$ strategy $\sigma'_i$ in $(\SB{\G}{h}, v)$ such that $\sigma'_i$ is finitely deviating from $\SB{\bar \sigma}{h}$ and 
\begin{equation}\label{finitelyDeviating}
\g_i(h\rho) > \g_i(h\rho'), 
\end{equation} 
where $\rho = \out{\SB{\bar \sigma}{h}}_v$ and $\rho' = \out{\sigma'_i, \SB{\sigma_{-i}}{h}}_v$. Let us consider such a strategy $\sigma'_i$ with a minimum number $n$ of deviation steps from $\out{\SB{\bar \sigma}{h}}_v$, and let $g_kv_k$, $1 \leq k \leq n$, be the histories in $\Hist_i(v)$ such that $\sigma'_i$ has a \devstep{g_kv_k} from $\out{\SB{\bar \sigma}{h}}_v$. Let us consider the subgame $(\SB{\G}{hg_n}, v_n)$ (see Figure~\ref{fig:weak-veryweak}). In this subgame, $\SB{\sigma'_i}{g_n}$ is not a profitable one-shot deviating strategy as $\bar \sigma$ is a very weak SPE. In other words, for $\varrho =  \out{\SB{\bar\sigma}{hg_n}}_{v_n}$ and $\varrho' = \out{\SB{\sigma'_i}{g_n}, \SB{\sigma_{-i}}{hg_n}}_{v_n}$, we have  
\begin{equation}\label{lastDeviation}
\g_i(hg_n\varrho) \leq \g_i(hg_n\varrho'). 
\end{equation}

\begin{figure}
\begin{center}
\begin{tikzpicture}[initial text=,auto, node distance=1cm, shorten >=1pt, scale=0.1] 

\node                                             (fictif)                             {};
\node[state, scale=0.5]                           (0)     [below=2cm of fictif]        {$v$};
\node                                            (0bbis)  [left=0.05cm of 0]            {};
\node[state, scale=0.5]                           (1)     [below=1cm of 0bbis]         {$v_1$};
\node[scale=0.6]                                  (1bis)  [below=2cm of 1]             {$h\rho$};
\node[fill=black, state, scale=0.1]               (3)     [below right=0.7cm of 1]     {};
\node                                            (1bbis)  [right=0.2cm of 3]           {};
\node[fill=black, state, scale=0.1]                           (4)     [below=1cm of 1bbis]         {};
\node[scale=0.6]                                  (2bis)  [below=2cm of 4]             {};
\node[fill=black, state, scale=0.1]               (5)     [below right=0.7cm of 4]     {};
\node                                          (fictif3)  [below right=0.1cm of 5]  {};
\node                                          (fictif4)  [below right=0.5cm of fictif3]  {};
\node                                            (3bis)  [below=2cm of 4]             {};
\node[state, scale=0.5]                           (7)     [below right=0.1cm of fictif4]     {$v_n$};
\node                                            (8)     [below=2.5cm of 7]         {};

\node  (fictif1)  at (-30, -65)           {};
\node [scale=0.6] (fictif2)  at (35, -65)          {$(\SB{\G}{h}, v)$};
\node (9)  at (34, -92)          {};
\node[fill=black, state, scale=0.1]               (10)     [below=0.5cm of 7]     {};
\node                                            (10bis)   [right=0.3cm of 10]           {};
\node  [scale=0.6]             (11)     [below=1.7cm of 10bis]     {$h\rho'$};
\node  (fictif5)  at (5, -110)           {};
\node [scale=0.6] (fictif6)  at (55, -110)          {$(\SB{\G}{h_n}, v_n)$};

\path[-]  (fictif)   edge      [out=-30,in=120]         node[scale=0.6]          {$h$}     (0)
 
          (0)   edge      [out=-100,in=100]         node[scale=0.6]          {$g_1$}     (1)

                edge                                                                     (fictif1)

                edge                                                                     (fictif2)

          (1)   edge      [->, dashed]             node[left, scale=0.6]    {}        (3)
				
				edge                               node[left, scale=0.6]    {}        (1bis)

          (4)   edge      [->, dashed]             node[left, scale=0.6]    {}        (5)

                edge                               node[left, scale=0.6]    {}        (3bis)

          (3)   edge      [out=-10,in=100]         node[right, scale=0.6]   {}        (4)

          (7)   edge             node[right, scale=0.6]   {}         (8)
 
                edge                                                                  (fictif5)

                edge                                                                  (fictif6)

	   	(fictif3)   edge      [dashed]                 (fictif4)

	      (10)   edge     [out=-35,in=100]         node[left, scale=0.6]    {}        (11);

\end{tikzpicture}
\end{center}
\caption{A strategy $\sigma'_i$ with a minimum number $n$ of deviation steps.}
\label{fig:weak-veryweak}
\end{figure}
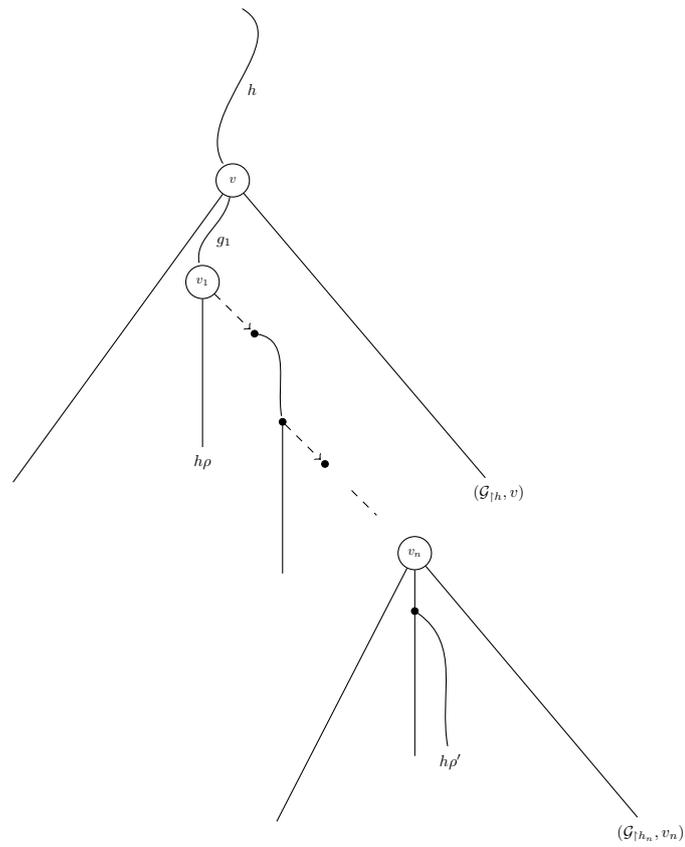

Notice that $n \geq 2$. Indeed, if $n = 1$, then $\rho = g_1\varrho$, $\rho' = g_1\varrho'$, and $\g_i(h\rho) \leq \g_i(h\rho')$ by (\ref{lastDeviation}). Therefore $\sigma'_i$ is not a profitable deviation in $(\SB{\G}{h}, v)$, in contradiction with its definition (\ref{finitelyDeviating}). We can thus construct a strategy $\tau'_i$ from $\sigma'_i$ such that these two strategies are the same except in the subgame $(\SB{\G}{hg_n}, v_n)$ where $\SB{\tau'_i}{g_n}$ and $\SB{\sigma_i}{hg_n}$ coincide. In other words $\tau'_i$ has $n-1$ deviation steps from $\out{\SB{\bar \sigma}{h}}_v$, that are exactly the \devstep{g_kv_k}s, $1 \leq k \leq n-1$, of $\sigma'_i$. Moreover, in the subgame $(\SB{\G}{h}, v)$, we have $\out{\tau'_i, \SB{\sigma_{-i}}{h}}_v = g_n\varrho$, and 
$$\g_i(hg_n\varrho) \leq \g_i(hg_n\varrho') < \g_i(h\rho)$$
by (\ref{finitelyDeviating}), (\ref{lastDeviation}), and $g_n\varrho' = \rho'$. It follows that $\tau'_i$ is a finitely deviating strategy that is profitable for player~$i$ in $(\SB{\G}{h}, v)$, with less deviation steps than $\sigma'_i$, a contradiction. This completes the proof of the first statement of Proposition~\ref{prop:weak-veryweak}.
For the second statement, it is enough to consider Example~\ref{ex:contrex}.
\end{proof}

Under the next hypotheses on the game or the costs, the equivalence between SPE, weak SPE, and very weak SPE holds. The first case, when the cost functions are continuous, is a classical result in game theory, see for instance  \cite{Fudenberg91}; the second case appears as a part of the proof of Kuhn's theorem \cite{kuhn53}.

\begin{proposition} \label{prop:upper}
Let $(\G, v_0)$ be an initialized game, and $\bar \sigma$ be a strategy profile. 
\begin{itemize}
\item If all cost functions $\g_i$ are continuous, or even upper-semicontinuous\footnote{In games where the players receive a payoff that they want to maximize, the hypothesis of upper-semicontinuity has to be replaced by lower-semicontinuity.}, then $\bar \sigma$ is an SPE iff $\bar \sigma$ is a weak SPE iff $\bar \sigma$ is a very weak SPE.
\item If $\G$ is a finite tree\footnote{In a finite tree game, the plays are finite sequences of vertices ending in a leaf and their cost is associated with 
the ending leaf. An example of such a game is depicted in Figure~\ref{TwoPlayerQuantitativeGame}.}, then $\bar \sigma$ is an SPE iff $\bar \sigma$ is a weak SPE iff $\bar \sigma$ is a very weak SPE. 
\end{itemize}
\end{proposition}

\begin{proof}
We only prove the first statement for cost functions $\g_i$ that are upper-semicontinuous. 
Let $\bar \sigma$ be a strategy profile in an initialized game $(\G, v_0)$. By Proposition~\ref{prop:weak-veryweak}, 
it remains to prove that if $\bar \sigma$ is a weak SPE, then it is an SPE, i.e., for each subgame $(\SB{\G}{h}, v)$, 
the strategy profile $\SB{\bar\sigma}{h}$ is a NE. Let $\sigma'_i$ be a player $i$ strategy in $(\SB{\G}{h}, v)$. 
If $\sigma'_i$ is finitely deviating, then it is not a profitable deviation for player~$i$ w.r.t. $\SB{\bar\sigma}{h}$ 
by hypothesis. Therefore, suppose that 
$$\rho' = \out{\sigma'_i,\SB{\sigma_{-i}}{h}}_v = g_1g_2 \ldots g_n \ldots$$ 
such that $\sigma'_i$ has a \devstep{h_n} from $\SB{\bar \sigma}{h}$ for all $n \geq 1$, where $h_0 = \epsilon$ and $h_n =  h_{n-1}g_n$. For each $n$, we define a finitely deviating strategy $\tau^n_i$ such that its deviation steps are the first $n$ deviation steps of $\sigma'_i$, that is, $\tau^n_i$ and $\sigma'_i$ are equal except in the subgame $(\SB{\G}{h_n}, \First(g_{n+1}))$ where $\SB{\tau^n_i}{h_n} = \SB{\sigma_i}{hh_n}$. By hypothesis, $\tau_i^n$ is not a profitable deviation, and thus for $\rho_n =  \out{\tau^n_i,\SB{\sigma_{-i}}{h}}_v$ we have
\begin{eqnarray} \label{eq:upper}
\g_i(h \ccdot \out{\SB{\bar\sigma}{h}}_v) \leq \g_i(h\rho_n).
\end{eqnarray}
As $\lim_{n \rightarrow \infty} h\rho_n =  h\rho'$ and $\g_i$ is upper-semicontinuous, 
we get $\limsup_{n \rightarrow \infty} \g_i(h\rho_n) \leq \g_i(h\rho')$. Therefore by (\ref{eq:upper}), $\g_i(h \ccdot \out{\SB{\bar\sigma}{h}}_v)  \leq \g_i(h\rho')$ showing that $\sigma'_i$ is not a profitable deviation for player~$i$ w.r.t. $\SB{\bar\sigma}{h}$.
\end{proof}

Recall that discounted sum games and quantitative reachability games are continuous. Thus for these games, the three notions of SPE, weak SPE and very weak SPE, are equivalent. 

\begin{corollary} \label{cor:SPEweakSPE}
Let $(\G, v_0)$ be an initialized quantitative reachability game, and $\bar \sigma$ be a strategy profile. Then $\bar \sigma$ is an SPE iff $\bar \sigma$ is a weak SPE iff $\bar \sigma$ is a very weak SPE.
\end{corollary}

On the opposite, the initialized game of Figure~\ref{fig:gameNoSPE} has a weak SPE but no SPE. Its cost function $\g_2$ is not upper-semicontinuous as $\lim_{n \rightarrow \infty} (v_0v_1)^nv_3^\omega = (v_0v_1)^\omega$ and $\lim_{n \rightarrow \infty} \g_2((v_0v_1)^nv_3^\omega)$ $=  1 > 0 = \g_2((v_0v_1)^\omega)$.

\section{Folk Theorems}\label{section:FolkTheorem}

\subsection{Folk Theorem for Weak SPEs}

In this section, we characterize in the form of a Folk Theorem the set of all outcomes of weak SPEs. Our approach is inspired\footnote{Our approach is however different.} by work~\cite{Flesch10} where a Folk Theorem is given for the set of outcomes of SPEs in games with cost functions that are upper-semicontinuous and have finite range. In this aim we define a nonincreasing sequence of sets of plays that initially contain all the plays, and then loose, step by step, some plays that for sure are not outcomes of a weak SPE, until finally reaching a fixpoint.

Let $(\G,v_0)$ be a game. For an ordinal $\alpha$ and a history $hv \in \Hist(v_0)$, let us consider the set $\PotPlay_\alpha(hv)= \{ \rho \mid \rho$ is  a \emph{potential} outcome of a weak NE in $(\SB{\G}{h},v)$ at step $\alpha \}$. This set is defined by induction on $\alpha$ as follows:

\begin{definition} \label{def:Palpha}
Let $(\G,v_0)$ be a quantitative game. The set $\PotPlay_\alpha(hv)$ is defined as follows for each ordinal $\alpha$ and history $hv \in \Hist(v_0)$:
\begin{itemize}
\item For $\alpha = 0$, 
\begin{eqnarray} \label{eq:P0}
\PotPlay_\alpha(hv) = \{ \rho \mid \rho \mbox{ is a play in } (\SB{\G}{h},v) \}.
\end{eqnarray}

\item For a successor ordinal $\alpha + 1$,
\begin{eqnarray} \label{eq:Palpha}
\PotPlay_{\alpha + 1}(hv) = \PotPlay_{\alpha}(hv) \setminus \ErasePlay_{\alpha}(hv)
\end{eqnarray}
such that $\rho \in \ErasePlay_{\alpha}(hv)$ (see Figure~\ref{fig:Ealpha}) iff 
\begin{itemize}
\item there exists a history $h'$, $hv \leq h' < h\rho$, and $\Last(h') \in V_i$ for some $i$,
\item there exists a vertex $v'$, $h'v' \not < h\rho$, 
\item such that $\forall \rho' \in \PotPlay_{\alpha}(h'v')$: $\g_i(h\rho) > \g_i(h'\rho')$. 
\end{itemize}

\item For a limit ordinal $\alpha$: 
\begin{eqnarray} \label{eq:Pbeta}
\PotPlay_\alpha(hv) = \bigcap_{\beta < \alpha} \PotPlay_\beta(hv).
\end{eqnarray}
\end{itemize}
\end{definition}

\begin{figure}[h!]
\begin{center}
\begin{tikzpicture}[initial text=,auto, node distance=1cm, shorten >=1pt, scale=0.8] 

\node[state, scale=0.5]               (0)                          {$v_0$};
\node[state, scale=0.5]               (1)    [below=2cm of 0]      {$v$};
\node[state, scale=0.5]               (2)    [below=1cm of 1]      {};
\node[state, scale=0.5]               (3)    [below right=of 2]    {$v'$};

\node[scale=0.6]                                        [right=0.5mm of 2]    {$\in V_i$};

\node[scale=0.6] (fictif) at (-1.5,-10)  {$\rho \in \ErasePlay_\alpha(hv)$};
\node[scale=0.6] (fictifbis) at (2.5,-10)  {$\forall \rho'$};
\node (fictif1) at (-5,-10)  {};
\node[scale=0.6] (fictif2) at (7,-10)  {$~~~~~~(\SB{\G}{h}, v)$};
\node (fictif1bis) at (0,-10)  {};
\node[scale=0.6] (fictif2bis) at (5,-10)  {$~~~~\PotPlay_\alpha(h'v')$};

\node (4) at (4,0.4) {};
\node (5) at (4,-5.3) {};

\path[-] (0)  edge      [out=-60,in=120]        node [scale=0.6]         {$h$}              (1)

          (1)  edge     [out=-80,in=100, thick]        node[midway, scale=0.6]  {}      (2)
			   edge                                                              (fictif1)
			   edge                                                              (fictif2)

          (2)  edge     [->, dashed]                node[left, scale=0.6]    {}      (3)

          (3)   edge    [out=-80,in=100, thick]                                         (fictifbis)
				edge                                                             (fictif1bis)
		   		edge                                                             (fictif2bis)

          (2)   edge    [out=-80,in=100, thick]                                         (fictif)

          (4)   edge    [<->]                    node [scale=0.6]    {$h'$}      (5);

\end{tikzpicture}
\end{center}
\caption{$\rho \in \ErasePlay_{\alpha}(hv)$.}
\label{fig:Ealpha}
\end{figure}

\noindent
Notice that an element $\rho$ of $\PotPlay_\alpha(hv)$ is a play in $(\SB{\G}{h},v)$ (and not in $(\G,v_0)$). Therefore it starts with vertex $v$, and $h \rho$ is a play in $(\G, v_0)$. For $\alpha + 1$ being a successor ordinal, play $\rho \in \ErasePlay_{\alpha}(hv)$ is erased from $\PotPlay_{\alpha}(hv)$ because for all $\rho' \in \PotPlay_\alpha(h'v')$, player~$i$ pays a lower cost $\g_i(h'\rho') < \g_i(h\rho)$, which means that $\rho$ is no longer a potentiel outcome of a weak NE in $(\SB{\G}{h},v)$.

The sequence $(\PotPlay_\alpha(hv))_\alpha$ is nonincreasing by definition, and reaches a fixpoint in the following sense.

\begin{proposition} \label{prop:fixpoint}
There exists an ordinal $\alpha_\ast$ such that $\PotPlay_{\alpha_\ast}(hv) = \PotPlay_{\alpha_\ast+1}(hv)$ for all histories $hv \in \Hist(v_0)$.
\end{proposition}

\begin{proof}
Let us fix a history $hv \in \Hist(v_0)$. The sequence $(\PotPlay_{\alpha}(hv))_\alpha$ reaches a fixpoint as soon as there exists $\alpha$ such that $\PotPlay_{\alpha}(hv) = \PotPlay_{\alpha+1}(hv)$. Indeed it follows that $\PotPlay_{\alpha +1}(hv) = \PotPlay_{\alpha+2}(hv)$ and then $\PotPlay_{\alpha}(hv) = \PotPlay_{\beta}(hv)$ for all $\beta > \alpha$. As the sequence $(\PotPlay_{\alpha}(hv))_\alpha$ is nonincreasing, this happens at the latest with $\alpha$ being equal to the cardinality of $\PotPlay_{0}(hv)$. Therefore with $\alpha_\ast = |V^\omega|$ being an ordinal greater than or equal to the cardinality of the set of all plays of $\G$, we get $\PotPlay_{\alpha_\ast}(hv) = \PotPlay_{\alpha_\ast+1}(hv)$ for all $hv \in \Hist(v_0)$. 
\end{proof}

\noindent
In the sequel, $\alpha_\ast$ always refers to the ordinal mentioned in Proposition~\ref{prop:fixpoint}. 

Our Folk Theorem for weak SPEs is the next one.

\begin{theorem} \label{thm:FolkThmWeakSPE}
Let $(\G,v_0)$ be a quantitative game.
There exists a weak SPE in $(\G,v_0)$ with outcome $\rho$ iff $\PotPlay_{\alpha_\ast}(hv) \neq \emptyset$ for all~$hv \in \Hist(v_0)$, and $\rho \in \PotPlay_{\alpha_\ast}(v_0)$.  
\end{theorem}

Before proving this theorem, we illustrate it with an example.

\begin{example} \label{ex:weakSPE}
Consider the example of Figure~\ref{fig:gameNoSPE}. Clearly, as $\PotPlay_0(hv_2)$ only contains the play $v_2^\omega$, then $\PotPlay_\alpha(hv_2) = \{v_2^\omega\}$ for all $\alpha$. Similarly $\PotPlay_\alpha(hv_3) = \{v_3^\omega\}$ for all $\alpha$. Let us detail the computation of $\PotPlay_\alpha(hv_0)$ and $\PotPlay_\alpha(hv_1)$.
\begin{itemize}
\item $\alpha = 0$.
For history $hv_0$, we have $\rho = (v_0v_1)^\omega \in \ErasePlay_0(hv_0)$, since $\g_1(h\rho) > \g_1(hv_0\rho')$ where $\rho' = v_2^\omega$ is the unique play of $\PotPlay_0(hv_0v_2)$. Similarly, for all $n \geq 1$, we have $\rho = (v_0v_1)^nv_0v_2^\omega \in \ErasePlay_0(hv_0)$, since $\g_2(h\rho) > \g_2(hv_0v_1\rho')$ where $\rho' = v_3^\omega$ is the unique play of $\PotPlay_0(hv_0v_1v_3)$. Thus $\ErasePlay_0(hv_0) = \{(v_0v_1)^\omega\} \cup (v_0v_1)^+v_0v_2^\omega$ and 
$$\PotPlay_1(hv_0) = \{v_0v_2^\omega\} \cup (v_0v_1)^+v_3^\omega.$$
For history $hv_1$, with the same kind of computations, we get $\ErasePlay_0(hv_1) = \{(v_1v_0)^\omega\} \cup (v_1v_0)^+v_2^\omega$ and 
$$\PotPlay_1(hv_1) =  v_1(v_0v_1)^\ast v_3^\omega.$$
\item $\alpha = 1$.
For history $hv_0$, we have $\rho = v_0v_2^\omega \in \ErasePlay_1(hv_0)$. Indeed $\g_1(h\rho) > \g_1(hv_0\rho')$ for all $\rho' \in \PotPlay_1(hv_0v_1) =  v_1(v_0v_1)^\ast v_3^\omega$. Notice that at the previous step, $\rho = v_0v_2^\omega \not\in \ErasePlay_0(hv_0)$. Indeed $\PotPlay_1(hv_0v_1) \subsetneq  \PotPlay_0(hv_0v_1)$, and $\g_1(h\rho) \leq \g_1(hv_0\rho')$ for $\rho' = (v_1v_0)^\omega \in \PotPlay_0(hv_0v_1)$.\footnote{This shows that as the sequence $(\PotPlay_\alpha(hv))_\alpha$ is nonincreasing for each $hv$, a play that is not removed from some $\PotPlay_\alpha(hv)$ can be removed later from some $\PotPlay_\beta(hv)$ with $\beta > \alpha$.} Play $v_0v_2^\omega$ is the only one that is removed from $\PotPlay_1(hv_0)$, and no play can be removed from $\PotPlay_1(hv_1)$. Therefore:
$$\PotPlay_2(hv_0) = (v_0v_1)^+v_3^\omega, \quad \PotPlay_2(hv_1) =  v_1(v_0v_1)^\ast v_3^\omega.$$
\item $\alpha = 2$. 
One checks that $\PotPlay_3(hv_0)  = \PotPlay_2(hv_0)$, and $\PotPlay_3(hv_1)  = \PotPlay_2(hv_1)$. Hence the fixpoint is reached with $\alpha_\ast = 2$, with $\PotPlay_{\alpha_\ast}(hv_0) = (v_0v_1)^+ v_3^\omega$, $\PotPlay_{\alpha_\ast}(hv_1) = v_1(v_0v_1)^\ast v_3^\omega$, $\PotPlay_{\alpha_\ast}(hv_2) = \{v_2^\omega\}$, and $\PotPlay_{\alpha_\ast}(hv_3) = \{v_3^\omega\}$. Therefore, the set of outcomes of weak SPEs in this game is equal to $(v_0v_1)^+ v_3^\omega$. The weak SPE depicted in Figure~\ref{fig:gameNoSPE} has outcome $v_0v_1v_3^\omega$.
\end{itemize}
\end{example}

The proof of Theorem~\ref{thm:FolkThmWeakSPE} follows from Lemmas~\ref{SPEToNotEmpty} and~\ref{NotEmptyToSPE}. 

\begin{lemma}{\label{SPEToNotEmpty}}
If $(\G, v_0)$ has a weak SPE $\bar\sigma$, then $\PotPlay_{\alpha_\ast}(hv) \neq \emptyset$ for all~$hv \in \Hist(v_0)$, and $\out{\bar\sigma}_{v_0} \in \PotPlay_{\alpha_\ast}(v_0)$.
\end{lemma}

\begin{proof}
Let us show, by induction on $\alpha$, that $\out{\SB{\bar\sigma}{h}}_{v} \in \PotPlay_{\alpha}(hv)$ for all~$hv \in \Hist(v_0)$. 

For $\alpha = 0$, we have $\out{\SB{\bar\sigma}{h}}_{v} \in \PotPlay_{\alpha}(hv)$ by definition of $\PotPlay_0(hv)$. 

Let $\alpha + 1$ be a successor ordinal. By induction hypothesis, we have that $\out{\SB{\bar\sigma}{h}}_{v} \in \PotPlay_{\alpha}(hv)$ for all $hv \in  \Hist(v_0)$. Suppose that there exists $hv$ such that $\out{\SB{\bar\sigma}{h}}_{v} \not\in \PotPlay_{\alpha + 1}(hv)$, i.e. $\out{\SB{\bar\sigma}{h}}_{v} \in \ErasePlay_{\alpha}(hv)$. This means that there is a history $h' = hg \in \Hist_i$ for some $i \in \Pi$ with $hv \leq h' < h\rho$, and there exists a vertex $v'$ with $h'v' \not < h\rho$, such that $\forall \rho' \in \PotPlay_{\alpha}(h'v')$, $\g_i(h \ccdot \out{\SB{\bar\sigma}{h}}_{v}) > \g_i(h'\rho')$. In particular, by induction hypothesis
\begin{eqnarray} \label{eq:profith'v'}
\g_i(h \ccdot \out{\SB{\bar\sigma}{h}}_{v}) > \g_i(h' \ccdot \out{\SB{\bar\sigma}{h'}}_{v'}).
\end{eqnarray} 
Let us consider the player $i$ strategy $\sigma'_i$ in  $(\SB{\G}{h},v)$ such that $g \ccdot \out{\SB{\bar\sigma}{h'}}_{v'}$ is consistent with $\sigma'_i$. Then $\sigma'_i$ is a finitely deviating strategy with the (unique) \devstep{g} from $\SB{\bar\sigma}{h}$. This strategy is a profitable deviation for player~$i$ in $(\SB{\G}{h},v)$ by (\ref{eq:profith'v'}), 
a contradiction with $\bar\sigma$ being a weak SPE. 

Let $\alpha$ be a limit ordinal. By induction hypothesis $\out{\SB{\bar\sigma}{h}}_{v} \in \PotPlay_{\beta}(hv), \forall \beta < \alpha$. Therefore $\out{\SB{\bar\sigma}{h}}_{v} \in \PotPlay_{\alpha}(hv) = \bigcap_{\beta < \alpha} \PotPlay_{\beta}(hv)$. 
\end{proof}

\begin{lemma}{\label{NotEmptyToSPE}}
Suppose that $\PotPlay_{\alpha_\ast}(hv) \neq \emptyset$ for all~$hv \in \Hist(v_0)$, and let $\rho \in \PotPlay_{\alpha_\ast}(v_0)$. Then $(\G, v_0)$ has a weak SPE with outcome $\rho$. 
\end{lemma}

\begin{proof}
We are going to show how to construct a very weak SPE $\bar\sigma$ (and thus a weak SPE by Proposition~\ref{prop:weak-veryweak}) with outcome $\rho$. The construction of $\bar\sigma$ is done step by step thanks to a progressive labeling of the histories $hv \in \Hist(v_0)$. Let us give an intuitive idea of the construction of $\bar\sigma$. Initially, we partially construct $\bar\sigma$ such that it produces an outcome in $(\G,v_0)$ equal to $\rho \in \PotPlay_{\alpha_\ast}(v_0)$; we also label each non-empty prefix of $\rho$ by $\rho$. Then we consider a shortest non-labeled history $h'v'$, and we correctly choose some $\rho' \in \PotPlay_{\alpha_\ast}(h'v')$  (we will see later how). We continue the construction of $\bar\sigma$ such that it produces the outcome $\rho'$ in $(\SB{\G}{h'},v')$, and for each non-empty prefix $g$ of $\rho'$, we label $h'g$ by $\rho'$ (notice that the prefixes of $h'$ have already been labeled by choice of $h'$). And so on. In this way, the labeling is a map $\gamma : \Hist(v_0) \rightarrow \bigcup_{hv} \,\PotPlay_{\alpha_\ast}(hv)$ that allows to recover from $h'g$ the outcome $\rho'$ of $\SB{\bar\sigma}{h'}$ in $(\SB{\G}{h'},v')$ of which $g$ is prefix. Let us now go into the details.

Initially, none of the histories is labeled. We start with history $v_0$ and the given play $\rho \in\PotPlay_{\alpha_\ast}(v_0)$. The strategy profile $\bar\sigma$ is partially defined such that $\out{\bar\sigma}_{v_0} = \rho$, that is, if $\rho = \rho_0 \rho_1 \ldots$, then $\sigma_i(\rho_{\leq n}) = \rho_{n+1}$ for all $\rho_n \in V_i$ and $i \in \Pi$. The non-empty prefixes $h$ of $\rho$ are all labeled with $\gamma(h) = \rho$.

At the following steps, we consider a history $h'v'$ that is not yet labeled, but such that $h'$ has already been labeled. By induction, $\gamma(h') = \out{\SB{\bar\sigma}{h}}_{v}$ such that $hv \leq h'$. Suppose that $\Last(h') \in V_i$, we then choose a play $\rho' \in \PotPlay_{\alpha_\ast}(h'v')$ such that (see Figure~\ref{fig:weakSPE})
\begin{eqnarray} \label{eq:chosen}
\g_i(h \ccdot \out{\SB{\bar\sigma}{h}}_{v}) \leq \g_i(h'\rho').
\end{eqnarray}

\begin{figure}

\begin{center}
\begin{tikzpicture}[initial text=,auto, node distance=1cm, shorten >=1pt, scale=0.8] 

\node[state, scale=0.5]               (0)                            {$v_0$};
\node[state, scale=0.5]               (1)    [below=1.2cm of 0]      {$v$};
\node[state, scale=0.5]               (2)    [below=1.2cm of 1]      {};
\node[state, scale=0.5]               (3)    [below right=of 2]      {$v'$};
\node[state, scale=0.5]               (4)    [below=0.6cm of 2]      {};

\node[scale=0.6]                             [right=0.5mm of 2]      {$\in V_i$};

\node[scale=0.6]    (fictif)       at (-1,-10)   {$\out{\SB{\bar\sigma}{h}}_v$};
\node[scale=0.6]    (fictifbis)    at (3,-10)    {$\rho' = \out{\SB{\bar\sigma}{h'}}_{v'}$};
\node               (fictif1bis)   at (1,-10)    {};
\node[scale=0.6]    (fictif2bis)   at (5,-10)    {$\PotPlay_{\alpha_\ast}(h'v')$};

\node (7) [right=2cm of 0] {};
\node (5) [above=0.1cm of 7]  {};
\node (6) [below=3.7cm of 5] {};

\path[-] (0)  edge     [out=-80,in=100]               node[midway, scale=0.6]  {$h$}      (1)
 		
	     (1)  edge     [out=-80,in=100, thick]        node[midway, scale=0.6]  {$g$}      (2)

         (2)  edge     [->, dashed]                       node[left, scale=0.6]    {}         (3)

         (3)   edge    [out=-80,in=100, thick]                                            (fictifbis) 
		   
		       edge                                                                       (fictif1bis)
		   	
		       edge                                                                       (fictif2bis)

         (2)   edge    [out=-80,in=90, thick]                                             (4)

         (4)   edge    [out=-80,in=100, thick]                                            (fictif)

         (5)  edge     [<->]                            node[scale=0.6]  {$h'$}      (6);

\end{tikzpicture}
\end{center}
\caption{Construction of a very weak SPE $\bar\sigma$.}
\label{fig:weakSPE}
\end{figure}
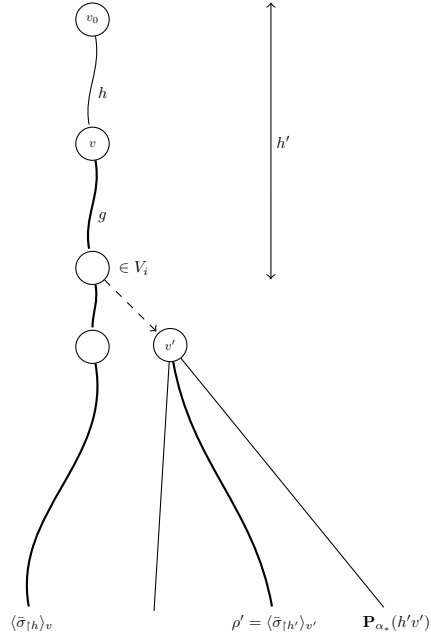

\noindent
Such a play $\rho'$ exists for the following reasons. By induction, we know that 
$\out{\SB{\bar\sigma}{h}}_{v} \in \PotPlay_{\alpha_\ast}(hv)$. Since $\PotPlay_{\alpha_\ast}(hv) = \PotPlay_{\alpha_\ast + 1}(hv)$ by Proposition~\ref{prop:fixpoint}, we have $\out{\SB{\bar\sigma}{h}}_{v} \not\in \ErasePlay_{\alpha_\ast}(hv)$, 
and we get the existence of $\rho'$ by definition of $\ErasePlay_{\alpha_\ast}(hv)$. We continue to construct 
$\bar\sigma$ such that $\out{\SB{\bar\sigma}{h'}}_{v'} = \rho'$, i.e. if $\rho' = \rho'_0 \rho'_1 \ldots$, 
then $\sigma_i(h'\rho'_{\leq n}) = h'\rho'_{n+1}$ for all $\rho'_n \in V_i$ and $i \in \Pi$. For all non-empty prefixes $g$ of $\rho'$, 
we define $\gamma(h'g) = \rho'$ (notice that the prefixes of $h'$ have already been labeled).  

Let us show that the constructed profile $\bar\sigma$ is a very weak SPE. Consider a history $hv \in \Hist_i$ for some $i \in \Pi$, and a one-shot 
deviating strategy $\sigma'_i$ from $\SB{\bar\sigma}{h}$ in the subgame $(\SB{\G}{h},v)$. 
Let $v'$ be such that $\sigma'_i(v) = v'$. By definition of $\bar\sigma$, 
we have $\gamma(hv) =  \out{\SB{\bar\sigma}{g}}_{u}$ for some history $gu \leq hv$ 
and $h \ccdot \out{\SB{\bar\sigma}{h}}_v = g \ccdot \out{\SB{\bar\sigma}{g}}_{u}$; and we have also $\gamma(hvv') = \out{\SB{\bar\sigma}{hv}}_{v'}$. Moreover $\g_i(g \ccdot \out{\SB{\bar\sigma}{g}}_{u})  \leq \g_i(hv \ccdot \out{\SB{\bar\sigma}{hv}}_{v'})$ by (\ref{eq:chosen}), and $\g_i(hv \ccdot \out{\SB{\bar\sigma}{hv}}_{v'})  = \g_i(h \ccdot \out{\sigma'_i,\SB{\sigma_{-i}}{h}}_{v})$ because $\sigma'_i$ is one-shot deviating. Therefore
$$\g_i(h \ccdot \out{\SB{\bar\sigma}{h}}_{v}) = \g_i(g \ccdot \out{\SB{\bar\sigma}{g}}_{u})  \leq \g_i(hv \ccdot \out{\SB{\bar\sigma}{hv}}_{v'})  = \g_i(h \ccdot \out{\sigma'_i,\SB{\sigma_{-i}}{h}}_{v})$$
which shows that $\SB{\bar\sigma}{h}$ is a very weak NE in $(\SB{\G}{h},v)$. Hence $\bar\sigma$ is a very weak SPE, and thus also a weak SPE.
\end{proof}

The next lemma will be useful in Sections~\ref{section:reach} and~\ref{section:prefixind}. It states that if a play $\rho$ belongs to $\PotPlay_\alpha(hv)$, then each of its suffixes $\rho_1$ also belongs to $\PotPlay_\alpha(hh_1v_1)$ such that $h_1\rho_1 = \rho$ and $v_1 = \First(\rho_1)$.

\begin{lemma} \label{lem:utile} 
Let $\rho \in \PotPlay_\alpha(hv)$. Then for all $h_1\rho_1 = \rho$, we have $\rho_1 \in \PotPlay_\alpha(hh_1v_1)$ with $v_1 = \First(\rho_1)$.
\end{lemma}

\begin{proof}
The proof is by induction on $\alpha$. The lemma trivially holds for $\alpha = 0$ by definition of $\PotPlay_0(hv)$. 

Let $\alpha + 1$ be a successor ordinal. Let $\rho \in \PotPlay_{\alpha +1}(hv)$ and $h_1\rho_1 = \rho$ with $v_1 = \First(\rho_1)$. As $\PotPlay_{\alpha +1}(hv) \subseteq \PotPlay_\alpha(hv)$, by induction hypothesis, we have $\rho_1 \in \PotPlay_\alpha(hh_1v_1)$. Suppose that $\rho_1 \in \ErasePlay_{\alpha}(hh_1v_1)$ (hence using a history $h'$ and a vertex $v'$ as in Definition~\ref{def:Palpha}). Then one can easily check by definition of $\ErasePlay_{\alpha}(hh_1v_1)$ that $\rho \in \ErasePlay_\alpha(hv)$ (by using the same $h'$ and $v'$), which is a contradiction with $\rho \in \PotPlay_{\alpha + 1}(hv)$. Therefore $\rho_1 \in \PotPlay_\alpha(hh_1v_1) \setminus \ErasePlay_{\alpha}(hh_1v_1) = \PotPlay_{\alpha+1}(hh_1v_1)$. 

Let $\alpha$ be a limit ordinal, and suppose that $\rho \in \PotPlay_\alpha(hv)$. As $\rho \in \PotPlay_\beta(hv)$ for all $\beta < \alpha$, we have $\rho_1 \in \PotPlay_\beta(hh_1v_1)$ by induction hypothesis. It follows that $\rho_1 \in \PotPlay_\alpha(hh_1v_1) = \bigcap_{\beta < \alpha}\PotPlay_\beta(hh_1v_1)$.
\end{proof}

\subsection{Folk Theorem for  SPEs}

In this section, as for weak SPEs, we characterize in the form of a Folk Theorem the set of all outcomes of SPEs. 
Nevertheless, we here need a more complex characterization with adapted sets $\PotPlay'_\alpha(hv)$, and this characterization only holds for cost functions that are upper-semicontinuous.\footnote{For games where players receive a payoff that they want to maximize, a similar Folk Theorem also exists for lower-semicontinuous cost functions.} The main difference appears in the definition of sets $\ErasePlay'_\alpha(hv)$ that will be used in place of $\ErasePlay_\alpha(hv)$. Indeed we will see that the set $\PotPlay_\alpha(h'v')$ of Figure~\ref{fig:Ealpha} has to be replaced by a more complex set $\ID_\alpha^{H,i}(h'v')$.

Let $(\G,v_0)$ be a game. For an ordinal $\alpha$ and a history $hv \in \Hist(v_0)$, as in the previous section, we consider the set $\PotPlay'_\alpha(hv) = \{ \rho \mid \rho$ is  a potential outcome of a NE\footnote{instead of a weak NE} in $(\SB{\G}{h},v)$ at step $\alpha \}$. In order to define these sets $\PotPlay'_\alpha(hv)$, we need to introduce new definitions. In Definition~\ref{def:devstep}, we have introduced the notion of deviation step of a strategy from a given strategy profile. We here propose another concept of deviation step in relation with two plays (see Figure~\ref{fig:devstepbis}). 

\begin{definition}
Let $h'v' \in \Hist(v_0)$ and $\rho' \in (\SB{\G}{h'},v')$. Let $h_1u_1, h_2u_2 \in \Hist(v_0)$ with $u_1, u_2 \in V$, and $\rho_1$ in $\PotPlay'_\alpha(h_1u_1)$. We say that $\rho'$ has a \emph{\devstep{h_2u_2}} from~$\rho_1$ if $h'v' \leq h_1u_1 < h_2u_2 < h' \rho'$ and $\prefix{h'\rho'}{h_1\rho_1} = h_2$. 
\end{definition}

\begin{figure}[h!]
\begin{center}
\begin{tikzpicture}[initial text=,auto, node distance=1cm, shorten >=1pt, scale=0.8] 

\node[state, scale=0.5]               (1)    [below=2cm of 0]      {$v_0$};
\node[state, scale=0.5]               (2)    [below=1cm of 1]      {$v'$};
\node[state, scale=0.5]               (5)    [below=0.8cm of 2]      {$u_1$};
\node[fill=black, state, scale=0.15]  (4)    [below=0.6cm of 5]    {};
\node[state, scale=0.5]               (3)    [below right=0.6cm of 4]    {$u_2$};

\node[scale=0.5] (fictif1) at (2,-2.8) {};
\node[scale=0.5] (fictif2) at (2,-6.3) {};

\node[scale=0.5] (fictif3) at (3,-2.8) {};
\node[scale=0.5] (fictif4) at (3,-7.75) {};

\node[scale=0.6] (fictif) at (-2,-12)  {$h_1\rho_1$};
\node[scale=0.6] (fictifbis) at (2,-12)  {$h'\rho'$};

\path[-]  (1)  edge             node[midway, scale=0.6]  {$h'$}      (2)

          (2)   edge                                                 (5)

          (5)   edge                                                 (4)

          (4)   edge    [out=-90,in=100]                      (fictif)

                edge    [->, dashed]                                 (3)

          (3)   edge    [out=-60,in=100]                      (fictifbis) 

          (fictif1)   edge     [<->]            node [scale=0.6]      {$h_1$}           (fictif2)

          (fictif3)   edge     [<->]            node [scale=0.6]      {$h_2$}            (fictif4);

\end{tikzpicture}
\end{center}	
\caption{Deviation step.}
\label{fig:devstepbis}
\end{figure}

Let us denote by $\IH(h'v')$ the set of all histories $h_2u_2$ such that $h'v' < h_2u_2$. Given a player~$i$ and $H \subseteq \IH(hv)$, the next definition introduces the notion of $(H,i)$-decomposition of a play $\rho'$. Such a play has a finite or infinite number of deviation steps such that for each {\devstep{h_nu_n}, the history $h_nu_n$ belongs to $H$. Figure~\ref{fig:Dalpha} illustrates the second case of Definition~\ref{def:Dalpha}, with the deviation steps highlighted with dashed edges. This definition also introduces the set $\ID_\alpha^{H,i}(h'v')$ composed of plays with a maximal $(H,i)$-decomposition.

\begin{figure}[h!] 	
\begin{center}
\begin{tikzpicture}[initial text=,auto, node distance=1cm, shorten >=1pt, scale=0.1] 

\node[state, scale=0.5]                           (0)                                     {$v_0$};
\node                                             (0bbis)    [right=0.2cm of 0]           {};
\node[state, scale=0.5]                           (1)        [below=1cm of 0bbis]         {$u_1$};
\node                                             (10)       [right=0.2cm of 1]           {};
\node[fill=black, state, scale=0.15]              (3bbis)    [below=1cm of 10]            {};
\node[scale=0.6]                                  (1bis)     [below=2cm of 3bbis]         {$\rho_1 \in \PotPlay'_\alpha(h_1u_1)$};
\node[state, scale=0.5]                           (3)        [below right=0.7cm of 3bbis] {$u_2$};
\node                                             (1bbis)    [right=0.2cm of 3]           {};
\node[fill=black, state, scale=0.15]              (4)        [below=1cm of 1bbis]         {};
\node[scale=0.6]                                  (2bis)     [below=2cm of 4]             {$\rho_2 \in \PotPlay'_\alpha(h_2u_2)$};
\node[state, scale=0.5]                           (5)        [below right=0.7cm of 4]     {$u_3$};
\node                                             (12)       [below right=0.1cm of 5]  {};
\node                                             (12bis)    [below right=0.7cm of 12]    {};
\node[state, scale=0.5]                           (7)        [below right=0.1cm of 12bis] {$u_{m}$};
\node                                             (8)        [right=0.2cm of 7]           {};
\node[fill=black, state, scale=0.15]              (9)        [below=1cm of 8]             {};
\node[scale=0.6]                                  (8bis)     [below=2cm of 9]             {$\rho_{m} \in \PotPlay'_\alpha(h_{m} u_{m})$};
\node[state, scale=0.5]                           (11)       [below right=0.7cm of 9, scale=0.8]     {$u_{m+1}$};
\node                                             (7bis)     [right=0.2cm of 11]           {};
\node[scale=0.6]                                  (8)        [below=3cm of 7bis]          {$\varrho' \in \PotPlay'_\alpha(h_{m+1} u_{m+1})$};

\node                [right=0.1cm of 3bbis, scale=0.6]     {$\in V_i$};
\node                [right=0.1cm of 4, scale=0.6]         {$\in V_i$};
\node                [right=0.1cm of 9, scale=0.6]         {$\in V_i$};

\node  (fictif)   [right=3cm of 0] {};
\node  (fictif1)  [above=0.1cm of fictif] {};
\node  (fictif2)  [below=2.8cm of fictif1] {};
\node  (fictifbis)  [right=4.0cm of 0] {};
\node  (fictif3)  [above=0.1cm of fictifbis] {};
\node  (fictif4)  [below=4.8cm of fictif3] {};
\node  (fictifbisbis)  [right=6.5cm of 0] {};
\node  (fictif5)  [above=0.1cm of fictifbisbis] {};
\node  (fictif6)  [below=8.1cm of fictif5] {};

\path[-]  (0)   edge      [out=-10,in=100]                node[scale=0.6]          {$h'$}        (1)

          (1)   edge      [out=-10,in=100, thick]         node[right, scale=0.6]   {$g_1$}     (3bbis)

          (4)   edge      [->, dashed]                    node[left, scale=0.6]    {}           (5)

                edge      [thick]                         node[left, scale=0.6]    {}           (2bis)

          (3)   edge      [out=-10,in=100, thick]         node[right, scale=0.6]   {$g_2$}      (4)

          (3bbis)	edge  [thick]                         node[left, scale=0.6]    {}           (1bis)
		
		     	edge      [->, dashed]                    node[left, scale=0.6]    {}           (3)

          %


          (12)  edge      [dashed]                        node[left, scale=0.6]    {}           (12bis)

          (7)   edge      [out=-10,in=100, thick]         node[right, scale=0.6]   {$g_{m}$}      (9)

          (9)   edge      [->, dashed]                    node[left, scale=0.6]    {}           (11)

		    	edge      [thick]                         node[left, scale=0.6]    {}           (8bis)
 
          (11)   edge     [out=-35,in=95, thick]          node[right, scale=0.6]   {}           (8)

          (fictif1)   edge      [<->]                     node[scale=0.6]          {$h_2$}        (fictif2)

          (fictif3)   edge      [<->]                     node[scale=0.6]          {$h_3$}        (fictif4)

          (fictif5)   edge      [<->]                     node[scale=0.6]          {$h_{m+1}$}        (fictif6);

\end{tikzpicture}
\end{center}	
\caption{$\rho' = g_1g_2 \ldots g_{m} \varrho'$ with a finite $(H,i)$-decomposition.}
\label{fig:Dalpha}
\end{figure}

%

\begin{definition} \label{def:Dalpha}
Let $h'v' \in \Hist(v_0)$ and $\rho' \in (\SB{\G}{h'},v')$. Let $i$ be a player, and $H \subseteq \IH(h'v')$. 
\begin{itemize}
\item $\rho'$ has an \emph{infinite $(H,i)$-decomposition}  $\rho' = g_1g_2 \ldots g_n \ldots$ if 
\begin{itemize}
\item for all $n \geq 1$, $\Last(g_n) \in V_i$
\item for all $n \geq 2$, $h_nu_n \in H$, and $\rho'$ has a \devstep{h_{n}u_{n}} from some $\rho_{n-1} \in \PotPlay'_\alpha(h_{n-1}u_{n-1})$
\end{itemize}
where $h_1 = h'$,  $h_{n+1} = h_ng_n$, and $u_n = \First(g_n)$ $\forall n \geq 1$. 
\item $\rho'$ has a \emph{finite $(H,i)$-decomposition}  $\rho' = g_1g_2 \ldots g_m \varrho'$ if
\begin{itemize}
\item for all $n$, $1 \leq n \leq m$,  $\Last(g_{n}) \in V_i$
\item for all $n$, $2 \leq n \leq m+1$, $h_nu_n \in H$, and $\rho'$ has a \devstep{h_nu_{n}} from some $\rho_{n-1} \in \PotPlay'_\alpha(h_{n-1}u_{n-1})$
\item $\varrho' \in \PotPlay'_\alpha(h_{m+1}u_{m+1})$
\end{itemize}
where $h_1 = h'$, $h_{n+1} = h_ng_n$, $u_n = \First(g_n)$ $\forall n, 1 \leq n \leq m$, and $u_{m+1} = \First(\varrho')$.
\end{itemize}
\noindent
We denote by $\ID_\alpha^{H,i}(h'v')$ the set of plays $\rho'$ with a \emph{maximal $(H,i)$-decomposition} in the following sense:
\begin{itemize}
\item $\rho'$ has an infinite $(H,i)$-decomposition, or
\item $\rho'$ has a finite $(H,i)$-decomposition $\rho' = g_1g_2 \ldots g_m \varrho'$, and there exists no $\rho''$ with a $(H,i)$-decomposition $\rho'' = g_1g_2 \ldots g_{m'}  \ldots$ or $\rho'' = g_1g_2 \ldots g_{m'}  \varrho''$ such that $\prefix{\rho'}{\rho''} = g_1g_2 \ldots g_{m+1}$.
\end{itemize}
\end{definition}


In the previous definition, $H$ can be chosen finite, infinite, or empty. If $H = \emptyset$, then for all $i \in \Pi$, for each $\rho' \in \ID_\alpha^{H,i}(h'v')$, $\rho'$ has no deviation step and thus $\rho' \in \PotPlay'_\alpha(h'v')$. This means that $\ID_\alpha^{H,i}(h'v') = \PotPlay'_\alpha(h'v')$ in this case. 

We are ready to define the sets $\PotPlay'_\alpha(hv)$ by induction on $\alpha$. The definition is similar to the one of $\PotPlay_\alpha(hv)$, except that when we erase $\rho \in \ErasePlay'_{\alpha}(hv)$ from $\PotPlay'_{\alpha}(hv)$, we use some set $\ID^{H,i}_{\alpha}(h'v')$ in place of $\PotPlay'_{\alpha}(h'v')$:

\begin{definition}
Let $(\G,v_0)$ be a quantitative game. The set $\PotPlay'_\alpha(hv)$ is defined as follows for each ordinal $\alpha$ and history $hv \in \Hist(v_0)$:
\begin{itemize}
\item For $\alpha = 0$, 
\begin{eqnarray} \label{eq:P'0}
\PotPlay'_\alpha(hv) = \{ \rho \mid \rho \mbox{ is a play in } (\SB{\G}{h},v) \}.
\end{eqnarray}

\item For a successor ordinal $\alpha + 1$,
\begin{eqnarray} \label{eq:P'alpha}
\PotPlay'_{\alpha + 1}(hv) = \PotPlay'_{\alpha}(hv) \setminus \ErasePlay'_{\alpha}(hv)
\end{eqnarray}
such that $\rho \in \ErasePlay'_{\alpha}(hv)$ iff 
\begin{itemize}
\item there exists a history $h'$, $hv \leq h' < h\rho$, and $\Last(h') \in V_i$ for some $i$,
\item there exists a vertex $v'$, $h'v' \not < h\rho$,
\item there exists $H \subseteq \IH(h'v')$,
\item such that $\forall \rho' \in \ID^{H,i}_{\alpha}(h'v')$: $\g_i(h\rho) > \g_i(h'\rho')$. 
\end{itemize}

(see Figure~\ref{fig:Ealpha} where $\ErasePlay_{\alpha}(hv)$ is replaced by $\ErasePlay'_{\alpha}(hv)$, and $\PotPlay'_{\alpha}(h'v')$ is replaced by $\ID^{H,i}_{\alpha}(h'v')$).

\item For a limit ordinal $\alpha$: 
\begin{eqnarray} \label{eq:P'beta}
\PotPlay'_\alpha(hv) = \bigcap_{\beta < \alpha} \PotPlay'_\beta(hv).
\end{eqnarray}
\end{itemize}
\end{definition}

Let us comment the case $\alpha + 1$. When $H= \emptyset$, we have $\ID^{H,i}_{\alpha}(h'v') = \PotPlay'_\alpha(h'v')$. Hence we recover the previous situation of weak SPEs. Using different sets $H \subseteq \IH(h'v')$ and $\ID^{H,i}_{\alpha}(h'v')$ allow to have sets $\ErasePlay'_{\alpha}(hv)$ bigger than $\ErasePlay_{\alpha}(hv)$, and thus more plays removed from $\PotPlay'_\alpha(hv)$ than in $\PotPlay_\alpha(hv)$. This situation will be illustrated in Example~\ref{ex:noSPE} hereafter.

The sequence $(\PotPlay'_\alpha(hv))_\alpha$ is nonincreasing by definition, and reaches a fixpoint in the following sense (the proof is the same as for Proposition~\ref{prop:fixpoint}).

\begin{proposition} \label{prop:fixpointSPE}
There exists an ordinal $\beta_\ast$ such that $\PotPlay'_{\beta_\ast}(hv) = \PotPlay'_{\beta_\ast+1}(hv)$ for all histories $hv \in \Hist(v_0)$.
\end{proposition}

Our Folk Theorem for SPEs is the next one. The second statement requires to work with upper-semicontinuous cost functions $\g_i$, $i \in \Pi$.

\begin{theorem} \label{thm:FolkSPE}
Let $(\G,v_0)$ be a quantitative game.
\begin{itemize}
\item If there exists an SPE in $(\G,v_0)$ with outcome $\rho$, then $\PotPlay'_{\beta_\ast}(hv) \neq \emptyset$ for all~$hv \in \Hist(v_0)$, and $\rho \in \PotPlay'_{\beta_\ast}(v_0)$.  
\item Suppose that all cost functions $\g_i$ are upper-semicontinuous. If $\PotPlay'_{\beta_\ast}(hv) \neq \emptyset$ for all~$hv \in \Hist(v_0)$, then there exists an SPE in $(\G,v_0)$ with outcome~$\rho$, for all $\rho \in \PotPlay'_{\beta_\ast}(v_0)$.
\end{itemize}
\end{theorem}

\begin{example} \label{ex:noSPE}
Before proving this theorem, we illustrate it with the example of Figure~\ref{fig:gameNoSPE} by showing that this game has no SPE (as stated in~\cite{SV03}). Let us compute the sets $\PotPlay'_{\alpha}(hv)$ and let us show that $\PotPlay'_3(hv_2) = \emptyset$. By the first statement of Theorem~\ref{thm:FolkSPE}, we will get that there is no SPE.

We have to do computations that are more complex than the ones of Example~\ref{ex:weakSPE}, due to the usage of sets $\ID^{H,i}_\alpha(h'v')$, with $H \subseteq \IH(h'v')$, instead of $\PotPlay_\alpha(h'v')$. Clearly, by definition $\PotPlay'_0(hv) = \PotPlay_0(hv)$ for all $hv$, and we have $\PotPlay'_\alpha(hv_2) = \{v_2^\omega\}$ and $\PotPlay'_\alpha(hv_3) = \{v_3^\omega\}$ for all $\alpha$ as in Example~\ref{ex:weakSPE}. 

Let us first illustrate Definition~\ref{def:Dalpha} with the computation of $\ID^{H,2}_0(v_0)$ with $H = \{v_0v_1v_0,$ $(v_0v_1)^2v_0\}$. The play $\rho'_1 = (v_0v_1)^2v_0v_2^\omega$ has a $(H,2)$-decomposition $g_1g_2\varrho'_1$ with $g_1 = g_2 = v_0v_1$ and $\varrho' = v_2^\omega$ (two deviation steps). Indeed $\rho'_1$ has a \devstep{v_0v_1v_0} from $v_0v_1v_3^\omega \in \PotPlay'_\alpha(v_0)$ and a \devstep{(v_0v_1)^2v_0} from $(v_0v_1)^2v_3^\omega \in \PotPlay'_\alpha(v_0v_1v_0)$, $\Last(g_1), \Last(g_2) \in V_2$, and $\varrho'_1 \in  \PotPlay'_\alpha((v_0v_1)^2v_0)$. This play $\rho'_1$ belongs to $\ID^{H,2}_0(v_0)$ because its $(H,2)$-decomposition uses the two possible steps of $H$. On the contrary $\rho'_2 = (v_0v_1)^2v_3^\omega$ has a $(H,2)$-decomposition $g_1\varrho'_2$ with $\varrho'_2 =  v_0v_1v_3^\omega$, and does not belong to $\ID^{H,2}_0(v_0)$ because $\prefix{\rho'_1}{\rho'_2} = g_1g_2$. One can check that $\ID^{H,2}(v_0) = \{ v_0v_2^\omega, 	v_0v_1v_0v_2^\omega \} \cup \{ (v_0v_1)^2\varrho' \mid \varrho' \in (\SB{\G}{(v_0v_1)^2},v_0) \}$. 

We can now detail the computation of $\PotPlay_\alpha'(hv_0)$ and $\PotPlay_\alpha'(hv_1)$.
\begin{itemize}
\item $\alpha = 0$.
For history $hv_0$, we have $\ErasePlay'_0(hv_0) \supseteq \ErasePlay_0(hv_0) = \{(v_0v_1)^\omega\} \cup (v_0v_1)^+v_0v_2^\omega$ because $\ID^{H,i}_0(h'v') = \PotPlay'_0(h'v') = \PotPlay_0(h'v')$ when $H = \emptyset$. In fact, one checks that $\ErasePlay'_0(hv_0) = \ErasePlay_0(hv_0)$, and thus 
$$\PotPlay'_1(hv_0) = \PotPlay_1(hv_0) = \{v_0v_2^\omega\} \cup (v_0v_1)^+v_3^\omega.$$
For instance, if we try to remove $v_0v_2^\omega$ with cost $(1,2)$ from $\PotPlay'_0(hv_0)$, we have to use some $H$ such that $\ID_0^{H,1}(hv_0v_1) \subseteq v_1(v_0v_1)^\ast v_3^\omega$ (with cost $(0,1)$). Such a $H$ does not exist. 
For history $hv_1$, we also have $\ErasePlay'_0(hv_1) = \ErasePlay_0(hv_1)$ and
$$\PotPlay'_1(hv_1) = \PotPlay_1(hv_1) =  v_1(v_0v_1)^\ast v_3^\omega.$$
\item $\alpha = 1$.
Again we have 
$$\PotPlay'_2(hv_0) = \PotPlay_2(hv_0) = (v_0v_1)^+v_3^\omega, \quad \PotPlay'_2(hv_1) = \PotPlay_2(hv_1) =  v_1(v_0v_1)^\ast v_3^\omega.$$
\item $\alpha = 2$. 
A difference appears at this step: $\PotPlay_3(hv_1) = v_1(v_0v_1)^\ast v_3^\omega$ whereas
$$\PotPlay'_3(hv_1) = \emptyset.$$
Indeed $\ErasePlay'_2(hv_1) =  \PotPlay'_2(hv_1)$. Consider for instance $\rho = v_1v_3^\omega \in \PotPlay'_2(hv_1)$, and $H = hv_1v_0(v_1v_0)^+ \subseteq \IH(hv_1v_0)$. Then $\ID^{H,2}_2(hv_1v_0)$ has a unique play $\rho' = (v_0v_1)^\omega$, and $\g_2(h\rho) > \g_2(hv_1\rho')$.
\end{itemize}
\end{example}

\begin{proof}[Proof of Theorem~\ref{thm:FolkSPE}] 
We begin by the first statement. Let $\bar\sigma$ be an SPE. As in the proof of Lemma~\ref{SPEToNotEmpty}, let us show by induction on $\alpha$ that $\out{\SB{\bar\sigma}{h}}_{v} \in \PotPlay'_{\alpha}(hv)$ for all~$hv \in \Hist(v_0)$. 

For $\alpha = 0$, we have $\out{\SB{\bar\sigma}{h}}_{v} \in \PotPlay'_{\alpha}(hv)$ by definition of $\PotPlay'_0(hv)$. 

Let $\alpha + 1$ be a successor ordinal. By induction hypothesis, we have that $\out{\SB{\bar\sigma}{h}}_{v} \in \PotPlay'_{\alpha}(hv)$ for all $hv$. Suppose that $\out{\SB{\bar\sigma}{h}}_{v} \not\in \PotPlay'_{\alpha + 1}(hv)$, i.e. $\out{\SB{\bar\sigma}{h}}_{v} \in \ErasePlay'_{\alpha}(hv)$. This means that there is a history $h' = hg \in \Hist_i$ for some $i$ with $hv \leq h' < h\rho$, there exists a vertex $v'$ with $h'v' \not < h\rho$, and there exists $H \subseteq \IH(h'v')$, such that $\forall \rho' \in \ID_{\alpha}^{H,i}(h'v')$, 
\begin{eqnarray} \label{eq:rho'}
\g_i(h \ccdot \out{\SB{\bar\sigma}{h}}_{v}) > \g_i(h'\rho').
\end{eqnarray}
Let us consider player~$i$ strategy $\sigma'_i$ in the subgame $(\SB{\G}{h},v)$ such that $\sigma'_i$ coincide with $\SB{\sigma_i}{h}$ except that 
\begin{eqnarray} \label{eq:h'v'}
\sigma'_i(h'_1) = v'_1, \mbox{ for all } hh'_1v'_1 \in H \cup \{h'v'\}.
\end{eqnarray}
Let $g \rho^\ast = \out{\sigma'_i,\SB{\sigma_{-i}}{h}}_{v}$. We get that $\rho^\ast$ has a maximal $(H,i)$-decomposition such that each $\rho_{n-1} \in \PotPlay'_{\alpha}(h_{n-1}u_{n-1})$ of Definition~\ref{def:Dalpha} is equal to $\out{\SB{\bar\sigma}{h_{n-1}}}_{u_{n-1}}$ (this play belongs to $\PotPlay'_{\alpha}(h_{n-1}u_{n-1})$ by induction hypothesis). Each deviation step of $\rho^\ast$ in the sense of Definition~\ref{def:Dalpha} corresponds to a deviation step of $\rho^\ast$ in the sense of Definition~\ref{def:devstep}\footnote{History $h'v'$ leads to one additional deviation step of $\rho^\ast$ in the sense of Definition~\ref{def:devstep} (see (\ref{eq:h'v'})).}. Moreover the $(H,i)$-decomposition of $\rho^\ast$ is finite (resp. infinite) iff $\sigma'_i$ is finitely (resp. infinitely) deviating from $\SB{\bar \sigma}{h}$. Thus this play $\rho^\ast$ belongs to  $\ID_{\alpha}^{H,i}(h'v')$, and by (\ref{eq:rho'}) we get $\g_i(h \ccdot \out{\SB{\bar\sigma}{h}}_{v}) > \g_i(h g \rho^\ast)$. Hence $\sigma'_i$ is a profitable deviation for player~$i$ in $(\SB{\G}{h},v)$, a contradiction with $\bar\sigma$ being an SPE. 

Let $\alpha$ be a limit ordinal. By induction hypothesis $\out{\SB{\bar\sigma}{h}}_{v} \in \PotPlay'_{\beta}(hv), \forall \beta < \alpha$. Therefore $\out{\SB{\bar\sigma}{h}}_{v} \in \PotPlay'_{\alpha}(hv) = \bigcap_{\beta < \alpha} \PotPlay'_{\beta}(hv)$. 
 
Let us now turn to the second statement of Theorem~\ref{thm:FolkSPE}. Let $\rho \in \PotPlay'_{\beta_\ast}(v_0)$. By Proposition~\ref{prop:upper}, it is enough to construct a very weak SPE $\bar\sigma$ with outcome $\rho$. The proof is very similar to the one of Lemma~\ref{NotEmptyToSPE}, where the construction of $\bar\sigma$ is done step by step thanks to a labeling~$\gamma$ of the histories. We briefly recall this proof and insist on the differences.
 
Initially, no history is labeled. We start with the play $\rho \in \PotPlay'_{\beta_\ast}(v_0)$, $\bar\sigma$ is partially defined such that $\out{\bar\sigma}_{v_0} = \rho$, and $\gamma(h) = \rho$ for all non-empty prefixes $h$ of~$\rho$.

At the following steps, let $h'v'$ be a history that is not yet labeled, but such that $h'$ has already been labeled. Suppose that $\Last(h') \in V_i$. By induction, $\gamma(h') = \out{\SB{\bar\sigma}{h}}_{v}$ such that $hv \leq h'$, and $\out{\SB{\bar\sigma}{h}}_{v} \in \PotPlay'_{\beta_\ast}(hv)$. Since $\PotPlay'_{\beta_\ast}(hv) = \PotPlay'_{\beta_\ast + 1}(hv)$ by Proposition~\ref{prop:fixpointSPE}, we have $\out{\SB{\bar\sigma}{h}}_{v} \not\in \ErasePlay'_{\beta_\ast}(hv)$. Therefore, with $H = \emptyset$ and $\ID^{H,i}_{\beta_\ast}(h'v') = \PotPlay'_{\beta_\ast}(h'v')$, we know that there exists  a play $\rho' \in \PotPlay'_{\beta_\ast}(h'v')$ such that $\g_i(h \ccdot \out{\SB{\bar\sigma}{h}}_{v}) \leq \g_i(h'\rho')$. Hence, we continue to construct $\bar\sigma$ such that $\out{\SB{\bar\sigma}{h'}}_{v'} = \rho'$, and all non-empty prefixes $g$ of $\rho'$ are labeled by $\gamma(h'g) = \rho'$. And so on. 

The constructed $\bar\sigma$ is a very weak SPE as in the proof of Lemma~\ref{NotEmptyToSPE}.
\end{proof}

The next proposition states that for cost functions that are upper-semicontinuous, sets $\PotPlay'_{\beta_\ast}(hv)$ and $\PotPlay_{\alpha_\ast}(hv)$ are all equal. This no longer the case as soon as one cost function is not upper-semicontinuous as shown by Examples~\ref{ex:weakSPE} and~\ref{ex:noSPE}.

\begin{proposition}
Let $(\G,v_0)$ be a quantitative game such that all its cost functions are upper-semicontinuous. Then for all $hv \in \Hist(v_0)$, $\PotPlay'_{\beta_\ast}(hv) = \PotPlay_{\alpha_\ast}(hv)$.
\end{proposition}

\begin{proof}
Let us prove by induction on $\alpha$ that $\PotPlay'_{\alpha}(hv) \subseteq \PotPlay_{\alpha}(hv)$ for all $hv$. These two sets are equal for $\alpha = 0$. Let $\alpha +1$ be a successor ordinal and suppose that $\PotPlay'_{\alpha}(hv) \subseteq \PotPlay_{\alpha}(hv)$. We have $\ErasePlay'_{\alpha}(hv) \supseteq \ErasePlay_{\alpha}(hv)$, and thus $\PotPlay'_{\alpha +1}(hv) \subseteq \PotPlay_{\alpha +1}(hv)$, because $\ID^{H,i}_{\alpha}(h'v') = \PotPlay'_{\alpha}(h'v')$ when $H = \emptyset$. For $\alpha$ being a limit ordinal, we easily have $\PotPlay'_{\alpha}(hv) \subseteq \PotPlay_{\alpha}(hv)$ by induction hypothesis.

Suppose now that $\PotPlay'_{\beta_\ast}(hv) \subsetneq \PotPlay_{\alpha_\ast}(hv)$ for some $hv$. Let $\rho \in \PotPlay_{\alpha_\ast}(hv) \setminus \PotPlay'_{\beta_\ast}(hv)$, and consider the initialized game $(\G',v'_0) = (\SB{\G}{h},v)$. Notice that the sets $\PotPlay_{\alpha}(h'v')$ and $\ErasePlay_{\alpha}(h'v')$ of this game $(\G',v'_0)$ are exactly the sets $\PotPlay_{\alpha}(hh'v')$ and $\ErasePlay_{\alpha}(hh'v')$ of $(\SB{\G}{h},v)$. By Theorem~\ref{thm:FolkThmWeakSPE}, there exists a weak SPE $\bar \sigma$ in $(\G',v'_0)$ with outcome $\rho$. Since the cost functions are upper-semicontinuous, $\sigma$ is also an SPE by Proposition~\ref{prop:upper}. Therefore, $\rho \in \PotPlay'_{\beta_\ast}(hv)$ by Theorem~\ref{thm:FolkSPE}, which is a contradiction.
%
\end{proof}

\section{Quantitative Reachability Games} \label{section:reach}

In this section, we focus on quantitative reachability games. Recall that in this case, the cost of a play for player $i$ is the number of edges to reach his target set of vertices $\Target_i$ (see Definition~\ref{def:reach}). Recall also that for quantitative reachability games, SPEs, weak SPEs, and very weak SPEs, are equivalent notions (see Corollary~\ref{cor:SPEweakSPE}).

It is known that there always exists an SPE in quantitative reachability games~\cite{BBD,Fudenberg83}.

\begin{theorem} \label{thm:ReachSPE}
Each quantitative reachability game $(\G,v_0)$ has an SPE. 
\end{theorem}

As SPEs and weak SPEs coincide in quantitative reachability games, we get the next result by Theorem~\ref{thm:FolkThmWeakSPE}.

\begin{corollary} \label{cor:existence}
Let $(\G, v_0)$ be a quantitative reachability game.
The sets $\PotPlay_{\alpha_\ast}(hv)$ are non-empty, for all $hv \in \Hist(v_0)$, and $\PotPlay_{\alpha_\ast}(v_0)$ is the set of outcomes of SPEs in $(\G,v_0)$.
\end{corollary}

The proof provided for Theorem~\ref{thm:ReachSPE} is non constructive since it relies on topological arguments.
Our main result is that one can algorithmically construct an SPE in $(\G,v_0)$ that is moreover finite-memory, thanks to the sets $\PotPlay_{\alpha_\ast}(hv)$.

\begin{theorem} \label{thm:main}
Each quantitative reachability initialized game $(\G,v_0)$ has a finite-memory SPE. Moreover there is an algorithm to construct such an SPE.
\end{theorem}

We can also decide whether there exists a (finite-memory) SPE such that the cost of its outcome is component-wise bounded by a given constant vector.

\begin{corollary} \label{cor:main}
Let $(\G,v_0)$ be a quantitative reachability initialized game, and let $\bar \const \in \IN^{|\Pi|}$ be a given $|\Pi|$-uple of integers. Then one can decide whether there exists  a (finite-memory) SPE $\bar \sigma$ such that $\g_i(\out{\bar \sigma}_{v_0}) \leq \const_i$ for all $i \in \Pi$.
\end{corollary}

The main ingredients of the proof of Theorem~\ref{thm:main} are the next ones; they will be detailed in the sequel of this section. We will give afterwards the proof of Corollary~\ref{cor:main}.

\begin{itemize}
\item Given $\alpha$, the infinite number of sets $\PotPlay_\alpha(hv)$ can be replaced by the finite number of sets $\PotS{I}{\alpha}{v}$ where $I$ is the set of players that did not reach their target set along history $h$.
\item The fixpoint of Proposition~\ref{prop:fixpoint} is reached with some natural number $\alpha_\ast \in \IN$.
\item Each $\PotS{I}{\alpha}{v}$ is a non-empty $\omega$-regular set, thus containing a ``lasso play" of the form $h \ccdot g^\omega$.
\item The lasso plays of each $\PotS{I}{\alpha}{v}$ allow to construct a finite-memory SPE.
\end{itemize}

The next lemma highlights a simple useful property of the cost functions $\g_i$ used in quantitative reachability games. The proof is immediate. 
\begin{lemma} \label{lem:costreach}
Let $i \in \Pi$ and $\rho \in \PotPlay_\alpha(hv)$. If player~$i$ did not reach his target set along history $h$, then $\g_i(h\rho) = \g_i(\rho) + |hv|$.
\end{lemma}

The next proposition is a key result that will be used several times later on. It states that it is impossible to have plays in $\PotPlay_\alpha(hv)$ with arbitrarily large costs for player $i$, without having a play in $\PotPlay_\alpha(hv)$ with an infinite cost for player $i$.

\begin{proposition} \label{prop:const}
Consider $\PotPlay_\alpha(hv)$ and $i \in \Pi$. If for all $\rho \in \PotPlay_\alpha(hv)$, we have $\g_i(\rho) < +\infty$, then there exists $\const$ such that for all $\rho \in \PotPlay_\alpha(hv)$, we have $\g_i(\rho) \leq \const$. The constant $\const$ only depends on $\PotPlay_\alpha(hv)$ and player~$i$.
\end{proposition}

\begin{proof}
Suppose that for all $n \in \IN$, there exists $\rho_n \in \PotPlay_\alpha(hv)$ such that $\g_i(\rho_n) > n$. By K\"onig's lemma, there exists $\rho = \lim_{k \rightarrow \infty} \rho_{n_k}$ for some subsequence $(\rho_{n_k})_k$ of $(\rho_{n})_n$. By definition of $\g_i$ in quantitative reachability games, we get $\g_i(\rho) = +\infty$. Let us prove by induction on $\alpha$ that $\rho \in \PotPlay_\alpha(hv)$; this  will establish Proposition~\ref{prop:const}.

Let $\alpha = 0$. As each $\rho_n \in \PotPlay_0(hv)$, then $\rho_n$  is a play in $(\SB{\G}{h},v)$ by definition of $\PotPlay_0(hv)$ (see (\ref{eq:P0}) in Definition~\ref{def:Palpha}). Therefore $\rho$ is also a play in $(\SB{\G}{h},v)$, and $\rho \in \PotPlay_0(hv)$.

Let $\alpha + 1$ be a successor ordinal. As for all $n$, $\rho_n \in \PotPlay_{\alpha +1}(hv) \subseteq \PotPlay_\alpha(hv)$, then $\rho \in \PotPlay_\alpha(hv)$ by induction hypothesis. Let us prove that $\rho \in \PotPlay_{\alpha +1}(hv)$. Suppose on the contrary that $\rho \in \ErasePlay_{\alpha}(hv)$ (see Definition~\ref{def:Palpha} and Figure~\ref{fig:Ealpha}). Then there exists a history $h' \in \Hist_j$ with $j \in \Pi$ and $hv \leq h' < h\rho$, there exists a vertex $v'$ with $h'v' \not < h\rho$, such that $\forall \rho' \in \PotPlay_{\alpha}(h'v')$: 
\begin{eqnarray} \label{eq:strict}
\g_j(h\rho) > \g_j(h'\rho').
\end{eqnarray}
It follows that player~$j$ did not reach his target set along $h'$. Hence by Lemma~\ref{lem:costreach}, we have $\g_j(h\rho) = \g_j(\rho) + |hv|$ and $\g_j(h'\rho') = \g_j(\rho') + |h'v'|$.
By (\ref{eq:strict}), $\g_j(\rho')$ is bounded. Hence by induction hypothesis with $\PotPlay_{\alpha}(h'v')$ and $j \in \Pi$, there exists a constant $\const$ such that $\g_j(\rho') \leq \const$, $\forall \rho' \in \PotPlay_{\alpha}(h'v')$.

\noindent
Suppose first that $\g_j(\rho) < +\infty$. Then, since $\rho = \lim_{k \rightarrow \infty} \rho_{n_k}$, it follows that for a large enough $n_k$, the plays $\rho$ and $\rho_{n_k}$ share a long common prefix on which player~$j$ reaches its target set, i.e. $\g_j(\rho) = \g_j(\rho_{n_k})$. It follows that with the same history $h'v'$ as above, by (\ref{eq:strict}) and Lemma~\ref{lem:costreach}, we have $\g_j(h\rho_{n_k})  > \g_j(h'\rho')$, $\forall \rho' \in \PotPlay_{\alpha}(h'v')$, showing that $\rho_{n_k} \in \ErasePlay_{\alpha}(hv)$, a contradiction.

\noindent
Suppose next that $\g_j(\rho) = +\infty$. Then, given $\const' = |h'v'| - |hv| + \const$, we can choose a large enough $n_k$ such that the plays $\rho$ and $\rho_{n_k}$ share a common prefix of length at least~$\const'$. Moreover, as $\g_j(\rho) = +\infty$, player~$j$ does not reach its target set along this prefix, i.e.  $\g_j(\rho_{n_k}) > \const'$. Therefore, using the same history $h'v'$ as above, by (\ref{eq:strict}) and Lemma~\ref{lem:costreach}, we have $\g_j(h\rho_{n_k}) > |hv| + \const' = |h'v'| + \const \geq \g_j(h'\rho')$, $\forall \rho' \in \PotPlay_{\alpha}(h'v')$. This shows that $\rho_{n_k} \in \ErasePlay_{\alpha}(hv)$, again a contradiction.

Let $\alpha$ be a limit ordinal. As for all $n$, $\rho_n \in \PotPlay_{\alpha}(hv) = \bigcap_{\beta < \alpha} \PotPlay_\beta(hv)$ (see (\ref{eq:Pbeta}) in Definition~\ref{def:Palpha}), then $\rho \in \PotPlay_\beta(hv), \forall \beta < \alpha$, by induction hypothesis. Hence $\rho \in \PotPlay_\alpha(hv) = \bigcap_{\beta < \alpha}\PotPlay_\beta(hv)$.

We have just shown that if all plays $\rho \in \PotPlay_\alpha(hv)$ have a cost $\g_i(\rho) < +\infty$, then there exists $\const$ such that $\g_i(\rho) \leq \const$ for all such $\rho$. This constant $\const$ depends on $\PotPlay_\alpha(hv)$ and player~$i$.
\end{proof}

\noindent
As a consequence of Proposition~\ref{prop:const}, we have that $\sup \{\g_i(\rho) \mid \rho \in \PotPlay_\alpha(hv)\}$ is equal to $\max \{\g_i(\rho) \mid \rho \in \PotPlay_\alpha(hv)\}$, and that this maximum belongs to $\IN \cup \{+\infty\}$.

\subsection{Sets $\PotS{I}{\alpha}{v}$} \label{subsec:PotSI}

Let $(\G, v_0)$ be a quantitative reachability game. Given a history $h = h_0 \ldots h_n$ in $(\G, v_0)$, we denote by $\Sh{I}{h}$ the set of players $i$ such that $\forall k$, $0 \leq k \leq n$, we have $h_k \not\in \Target_i$. In other words $\Sh{I}{h}$ is the set of players that did not reach their target set along history $h$. If $h$ is empty, then $\Sh{I}{h} = \Pi$. The next lemma indicates that sets $\PotPlay_\alpha(hv)$ only depend on $v$ and  $\Sh{I}{h}$, and thus not on $h$ (we do no longer take care of players that have reached their target set along~$h$).

\begin{lemma} \label{lem:hversI}
For $h_1v, h_2v \in \Hist(v_0)$, if $\Sh{I}{h_1} = \Sh{I}{h_2}$, then $\PotPlay_\alpha(h_1v) = \PotPlay_\alpha(h_2v)$ for all~$\alpha$. 
\end{lemma}

\begin{proof}
The proof is by induction on $\alpha$. By definition, we have $P_0(h_1v) = P_0(h_2v)$. 

Suppose that $\alpha + 1$ is a successor ordinal. By induction hypothesis, $\PotPlay_\alpha(h_1v) =  \PotPlay_\alpha(h_2v)$. 
Let us prove that $\ErasePlay_\alpha(h_1v) = \ErasePlay_\alpha(h_2v)$ which will imply that $\PotPlay_{\alpha +1}(h_1v) =  \PotPlay_{\alpha +1}(h_2v)$. If $\rho \in \ErasePlay_\alpha(h_1v)$, 
it means that there exists a history $h'_1 = h_1g \in \Hist_i$ with $h'_1 < h_1\rho$, 
there exists a vertex $v'$ with $h'_1v' \not < h_1\rho$, such that $\forall \rho' \in \PotPlay_{\alpha}(h'_1v')$, 
we have $\g_i(h_1\rho) > \g_i(h'_1\rho')$, i.e.
\begin{eqnarray}\label{eq:profit}
\g_i(\rho) > \g_i(g\rho'). 
\end{eqnarray}
In particular, $i \in \Sh{I}{h'_1}$. Let us consider the history $h'_2 = h_2g$. By hypothesis, $\Sh{I}{h_1} = \Sh{I}{h_2}$, 
and therefore $\Sh{I}{h'_1} = \Sh{I}{h'_2}$ and $i \in \Sh{I}{h'_2}$. Thus by induction hypothesis 
$\PotPlay_{\alpha}(h'_1v') = \PotPlay_{\alpha}(h'_2v')$. It follows that for $\forall \rho' \in \PotPlay_{\alpha}(h'_2v')$, 
we have $\g_i(h_2\rho) > \g_i(h_2g\rho') = \g_i(h'_2\rho')$ by (\ref{eq:profit}), and then $\rho \in \ErasePlay_\alpha(h_2v)$. 
Symmetrically, if $\rho \in \ErasePlay_\alpha(h_2v)$, then $\rho \in \ErasePlay_\alpha(h_1v)$. 
We can conclude that $\PotPlay_{\alpha+1}(h_1v) = \PotPlay_{\alpha+1}(h_2v)$.

Suppose that $\alpha$ is a limit ordinal. As $\PotPlay_\alpha(h_1v) = \bigcap_{\beta < \alpha} \PotPlay_\beta(h_1v)$, and $\PotPlay_\beta(h_1v) =  \PotPlay_\beta(h_2v)$ by induction hypothesis, it follows that $\PotPlay_\alpha(h_1v) =  \PotPlay_\alpha(h_2v)$.
\end{proof}

Thanks to this lemma, we can introduce the next definitions.

\begin{definition}
Let $(\G, v_0)$ be a quantitative reachability initialized game. Let $I \subseteq \Pi$ be such that $I = \Sh{I}{h}$ for some $h \in \Hist(v_0)$. We denote by 
\begin{itemize}
\item $\PotS{I}{\alpha}{v}$ the set $\PotPlay_\alpha(hv)$, and by
\item $\EraseS{I}{\alpha}{v}$ the set $\ErasePlay_\alpha(hv)$.
\end{itemize}
\end{definition}

\noindent 
In particular, $\PotS{\Pi}{\alpha}{v_0} = \PotPlay_\alpha(v_0)$ and $\EraseS{\Pi}{\alpha}{v_0} = \ErasePlay_\alpha(v_0)$.


Given $\alpha$, the infinite number of sets $\PotPlay_\alpha(hv)$ can thus be replaced by the finite number of sets $\PotS{I}{\alpha}{v}$.
Moreover, Proposition~\ref{prop:const} can be rephrased as follows.

\begin{corollary} \label{cor:const}
Consider $\PotS{I}{\alpha}{v}$ and $i \in I$. If for all $\rho \in \PotS{I}{\alpha}{v}$, we have $\g_i(\rho) < +\infty$, then there exists $\const$ such that for all $\rho \in \PotS{I}{\alpha}{v}$, we have $\g_i(\rho) \leq \const$. The constant $\const$ only depends on $\alpha, I, v$, and $i$.
\end{corollary}

\begin{proof} 
Let $h$ be such that $I = \Sh{I}{h}$. Consider $\PotPlay_\alpha(hv) = \PotS{I}{\alpha}{v}$, and $i \in I$. By Proposition~\ref{prop:const}, if for all $\rho \in \PotPlay_\alpha(hv)$, $\g_i(\rho) < +\infty$, then there exists $\const$ (depending on $\PotPlay_\alpha(hv)$ and $i$) such that for all $\rho \in \PotPlay_\alpha(hv)$, $\g_i(\rho) \leq \const$. By Lemma~\ref{lem:hversI}, $c$ depends on $\alpha$, $I$, $v$, and $i$. 
\end{proof}

As a consequence of Corollary~\ref{cor:const}, we give the next definition that indicates the maximum costs for plays in $\PotS{I}{\alpha}{v}$.
\begin{definition} \label{def:costmax}
Given $\PotS{I}{\alpha}{v}$, we define $\barmaxg(\PotS{I}{\alpha}{v})$ such that
$$\maxg_i(\PotS{I}{\alpha}{v}) = \left\{\begin{array}{ll}  -1 &\mbox{if } i \not\in I, \\
                                               \max \{ \g_i(\rho) \mid \rho \in \PotS{I}{\alpha}{v} \} &\mbox{if } i \in I.
\end{array}\right .$$
\end{definition}

\noindent 
In this definition, $-1$ indicates that player $i$ has already visited his target set $T_i$, and the $\max$ belongs to $\IN \cup \{+\infty\}$.

\subsection{Fixpoint with $\alpha_\ast \in \IN$} \label{subsec:fixpoint}

In this section, we aim at proving that the fixpoint (when computing the sets $\PotS{I}{\alpha}{v}$, see Proposition~\ref{prop:fixpoint}) is reached in a finite number of steps, that is $\alpha_\ast \in \IN$. 

We first need to introduce some notions about the sets $\PotS{I}{\alpha}{v}$. Let $\rho = \rho_0\rho_1 \ldots \in \PotS{I}{\alpha}{v}$. We use a map $\chi$ that \emph{decorates} each $\rho_n$ by some set $J \subseteq \Pi$. The aim of the decoration $\chi(\rho_n)$ is to indicate at vertex $\rho_n$, which players of $I$ did not reach their target set along $\rho_{<n}$. More precisely, $\chi(\rho_n) = I \cap \Sh{I}{\rho_{<n}}$. In particular $\chi(\rho_0) = I \cap \Pi = I$.


Let $\PotS{I}{\alpha}{v}$ and $(v,v') \in E$. We now adapt Definition~\ref{def:costmax} to mention the maximum costs for plays in $\PotS{I}{\alpha}{v}$ \emph{starting with edge $(v,v')$}. We define $\barmaxg(\PotS{I}{\alpha}{v},v')$ as follows:
\begin{eqnarray} \label{eq:adapt}
\maxg_i(\PotS{I}{\alpha}{v},v') = \left\{\begin{array}{ll}  -1 &\mbox{if } i \not\in I, \\
                                               \max \{ \g_i(\rho) \mid \rho \in \PotS{I}{\alpha}{v} \mbox{ and } \rho_0\rho_1 = vv' \} &\mbox{if } i \in I. 
\end{array}\right .
\end{eqnarray}

\noindent
In this definition, the $\max$ is equal to -1 when the set $\{ \g_i(\rho) \mid \rho \in \PotS{I}{\alpha}{v} \mbox{ and } \rho_0\rho_1 = vv' \}$ is empty.\footnote{Notice that as $\PotS{I}{\alpha}{v}$ is non-empty, there exists some $(v,v') \in E$ such that this set is non-empty.}

The sequence $(\barmaxg(\PotS{I}{\alpha}{v},v'))_\alpha$ is nonincreasing for the usual component-wise ordering over $(\IN \cup \{-1,+\infty\})^{\Pi}$ since $\PotS{I}{\alpha}{v}$ is nonincreasing for the inclusion by definition. Therefore it reaches a fixpoint that we want to relate to the fixpoint $\PotS{I}{\alpha_\ast}v$ of Proposition~\ref{prop:fixpoint}. This is done in the following lemma.

\begin{lemma} \label{lem:lien}
\begin{itemize}
\item If $\PotS{I}{\alpha}{v} = \PotS{I}{\alpha+1}{v}$, then for all $(v,v') \in E$, $\barmaxg(\PotS{I}{\alpha}{v},v') = \barmaxg(\PotS{I}{\alpha +1}{v},v')$.
\item If $\PotS{I}{\alpha}{v} \neq \PotS{I}{\alpha+1}{v}$, then there exist $J \subseteq \Pi$ and $(u,u') \in E$ such that $\barmaxg(\PotS{J}{\alpha}{u},u') \neq \barmaxg(\PotS{J}{\alpha +1}{u},u')$.
\end{itemize} 
\end{lemma}

\begin{proof}
The first statement is immediate from definition of $\barmaxg$. Let us prove the second statement. Consider $\rho  = \rho_0\rho_1 \ldots \in \EraseS{I}{\alpha}{v}$. Then there exist $i \in \Pi$, $n \in \IN$ and $v' \neq \rho_{n+1}$ with $\rho_n \in V_i$, $\chi(\rho_n) = J$,  $\chi(\rho_{n+1}) = J'$, such that $\forall \rho' \in \PotS{J'}{\alpha}{v'}$ we have $\g_i(\rho) > \g_i(\rho_0\ldots\rho_n\rho')$ or equivalently (by Lemma~\ref{lem:costreach})
\begin{eqnarray}\label{eq:profitnui}
\g_i(\rho) - (n+1) > \g_i(\rho'). 
\end{eqnarray}
(see the definition of $\ErasePlay_\alpha(hv)$ with $\Sh{I}{h} = I$ in Definition~\ref{def:Palpha} and Figure~\ref{fig:decore}). Notice that $i \in J'$. 
\begin{figure}[h!]
\begin{center}
\begin{tikzpicture}[initial text=,auto, node distance=1cm, shorten >=1pt, scale=0.8] 

\node[state, scale=0.5]               (1)    [below=2cm of 0]      {$v$};
\node[state, scale=0.5]               (2)    [below=1cm of 1]      {$u$};
\node[state, scale=0.5]               (3)    [below right=of 2]    {$v'$};
\node[state, scale=0.5]               (4)    [below=0.6cm of 2]      {$u'$};

\node[scale=0.6]                                        [right=0.5mm of 2]    {$\in V_i$};
\node[scale=0.6]                                        [left=0.5mm of 1]    {$I$};
\node[scale=0.6]                                        [left=0.5mm of 2]    {$J$};
\node[scale=0.6]                                        [left=0.5mm of 4]    {$J'$};

\node[scale=0.6] (fictif) at (-1,-10)  {$\rho \in \EraseS{I}{\alpha}{v}$};
\node[scale=0.6] (fictifbis) at (3,-10)  {$\forall \rho'$};
\node (fictif1bis) at (1,-10)  {};
\node[scale=0.6] (fictif2bis) at (5,-10)  {$\PotS{J'}{\alpha}{v'}$};

\path[-]  (1)  edge     [out=-80,in=100, thick]        node[midway, scale=0.6]  {}      (2)

          (2)  edge     [->, dashed]                node[left, scale=0.6]    {}      (3)

          (3)   edge    [out=-80,in=100, thick]                                         (fictifbis) 
				edge                                                             (fictif1bis)
		   		edge                                                             (fictif2bis)

          (2)   edge    [out=-80,in=90, thick]                                         (4)

          (4)   edge    [out=-80,in=100, thick]                                         (fictif);

\end{tikzpicture}
\end{center}
\caption{$\rho \in \EraseS{I}{\alpha}{v}$, with $\rho_n = u$ and $\rho_{n+1} = u'$.}
\label{fig:decore}
\end{figure}
Let us prove that $\barmaxg(\PotS{J}{\alpha}{u},u') \neq \barmaxg(\PotS{J}{\alpha +1}{u},u')$ with $u = \rho_n$ and $u' = \rho_{n+1}$. As $\rho \in \PotS{I}{\alpha}{v}$, then $\rho_{\geq n} \in \PotS{J}{\alpha}{u}$ by Lemma~\ref{lem:utile}. As $\g_i(\rho_{\geq n}) = \g_i(\rho) - n$, this implies that 
\begin{eqnarray}\label{eq:profitmax}
\maxg_i(\PotS{J}{\alpha}{u},u') \geq \g_i(\rho) - n.
\end{eqnarray}
Let $\varrho \in \PotS{J}{\alpha}{u}$ be such that $\varrho$ starts with edge $(u,u')$ and has maximal cost $\maxg_i(\PotS{J}{\alpha}{u},u')$. One gets 
$$\g_i(\varrho) = \maxg_i(\PotS{J}{\alpha}{u},u') \geq \g_i(\rho) - n > \g_i(\rho_n\rho')$$
by (\ref{eq:profitnui}) and (\ref{eq:profitmax}). 
By considering the set $\PotS{J'}{\alpha}{v'}$ in Figure~\ref{fig:decore}, it follows that $\varrho \in \EraseS{J}{\alpha}{u}$ for all such plays $\varrho$. Hence $\PotS{J}{\alpha +1}{u} \subsetneq \PotS{J}{\alpha}{u}$ and $\maxg_i(\PotS{J}{\alpha +1}{u},u') < \maxg_i(\PotS{J}{\alpha}{u},u')$. This completes the proof.
\end{proof}

We are now able to prove that the ordinal $\alpha_\ast$ of Proposition~\ref{prop:fixpoint} is an integer.

\begin{corollary} \label{cor:integer}
There exists an integer $\alpha_\ast$ such that $\PotS{I}{\alpha_\ast}{v} = \PotS{I}{\alpha_\ast +1}{v}$ for all $v \in V$ and $I \subseteq \Pi$.
\end{corollary}

\begin{proof}
Notice that there is a finite number of sequences $(\barmaxg(\PotS{I}{\alpha}{v},v'))_\alpha$ since they depend on $I \subseteq \Pi$ and $(v,v') \in E$. As the component-wise ordering over $(\IN \cup \{-1,+\infty\})^{\Pi}$ is a well-quasi-ordering and all these sequences are nonincreasing, there exists an integer (and not only an ordinal) $\alpha'_\ast$ such that $\barmaxg(\PotS{I}{\alpha'_\ast}{v},v') = \barmaxg(\PotS{I}{\alpha'_\ast +1}{v},v')$ for all $I \subseteq \Pi$ and $(v,v') \in E$. By Lemma~\ref{lem:lien}, we get that $\alpha_\ast \leq \alpha'_\ast$, showing that $\alpha_\ast \in \IN$.
\end{proof}

\subsection{The sets $\PotS{I}{\alpha}{v}$ are $\omega$-regular} \label{subsec:regular}

In this section, we prove that each set $\PotS{I}{\alpha}{v}$ is $\omega$-regular. Instead of providing the construction of a B\"uchi automaton (which would lead to many technical details), we prefer to show that each set $\PotS{I}{\alpha}{v}$ is MSO-definable. It is well-known that a set of $\omega$-words is $\omega$-regular iff it is MSO-definable, by B\"uchi theorem \cite{Thomas90}. Moreover from the B\"uchi automaton, one can construct an equivalent MSO-sentence, and conversely. One can also decide whether an MSO-sentence is satisfiable \cite{Thomas90}. We recall that MSO-logic uses:
\begin{itemize}
\item variables $x, y, \ldots$ ($X,Y, \ldots$ resp.) to describe a position (a set of positions resp.) in an $\omega$-word $\rho$, and relations $X(x)$ to mention that $x$ belongs to~$X$,
\item relations $Q_u(x)$, $u \in V$, to mention that such vertex $u$ is at position $x$ of $\rho$,
\item relations $x < y$ and $Succ(x,y)$ to mention that position $y$ is after position $x$, and position $y$ is successor of position $x$ respectively,
\item connectives $\vee, \wedge, \neg$ and quantifiers $\exists x$, $\forall x$, $\exists X$, $\forall X$
\end{itemize}
\noindent
Recall that constants $0, 1, \ldots$ are definable. We will use notation $x+1$ (and more generally $x+c$, with $c$ a constant) instead of $Succ(x,y)$.

\begin{proposition} \label{prop:regular}
Each $\PotS{I}{\alpha}{v}$ is an $\omega$-regular set.
\end{proposition}

We begin with a lemma that states that if $\PotS{I}{\alpha}{v}$ is $\omega$-regular, then the maximum of its costs is computable.

\begin{lemma} \label{lem:costmaxcomput}
If $\PotS{I}{\alpha}{v}$ is MSO-definable, then $\barmaxg(\PotS{I}{\alpha}{v})$ is computable.
\end{lemma}

\begin{proof}
Before proving this lemma, we need to establish two properties. The first one states that one can decide whether $\PotS{I}{\alpha}{v}$ has a play $\rho$ with a given cost for a given player. The second one states that when $\maxg_i(\PotS{I}{\alpha}{v})$ is finite, then this number is bounded by the number of states of a B\"uchi automaton accepting $\PotS{I}{\alpha}{v}$.

\emph{(i)} Let $\const \in \IN \cup \{+\infty\}$ and $i \in I$. Let $\phi$ be an MSO-sentence 
defining $\PotS{I}{\alpha}{v}$. Let us show that one can decide whether $\PotS{I}{\alpha}{v}$ 
has a play $\rho$ with cost $\g_i(\rho) = \const$. There exists an MSO-sentence $\varphi$ 
expressing that $\g_i(\rho) = \const$. Indeed, if $\const = +\infty$, then $\varphi$ is the 
sentence $\forall x \cdot \neg (\vee_{u \in \Target_i} Q_u(x))$, and if $\const < +\infty$, 
it is the sentence $(\forall x < \const \cdot \neg (\vee_{u \in \Target_i} Q_u(x))) \wedge (\vee_{u \in \Target_i} Q_u(c))$. 
Therefore one can decide whether the MSO-sentence $\phi \wedge \varphi$ is satisfiable by some play $\rho$.

\emph{(ii)} Let $i \in I$ and suppose that $\maxg_i(\PotS{I}{\alpha}{v}) < +\infty$. Let $\cal B$ be a B\"uchi automaton accepting $\PotS{I}{\alpha}{v}$. We now show that $\maxg_i(\PotS{I}{\alpha}{v}) < n$ where $n$ is the number of states of $\cal B$. Assume the contrary and consider an accepting run $r = r_0r_1 \ldots $ of $\cal B$ on a play $\rho = \rho_0\rho_1\ldots \in \PotS{I}{\alpha}{v}$ with $\g_i(\rho) = \maxg_i(\PotS{I}{\alpha}{v})  \geq n$. The prefix $r_{\leq n}$ of $r$ has a cycle $r_k \ldots r_l$ with $0 \leq k < l \leq n$ and $r_k = r_l$. This cycle can be repeated once, while keeping an accepting run labeled by $\rho' = \rho_0 \ldots (\rho_k \ldots \rho_{l-1})^2 \rho_{\geq l}$. As $\g_i(\rho) \geq n$, it follows that $\g_i(\rho') = \maxg_i(\PotS{I}{\alpha}{v}) + (l-k)$. Therefore we get a contradiction with $\g_i(\rho) = \maxg_i(\PotS{I}{\alpha}{v})$.

Let us prove the lemma. By definition $\maxg_i(\PotS{I}{\alpha}{v})$ equals $-1$ if $i \not\in I$, and is thus computable in this case. Let $i \in I$. By \emph{(i)}, one can decide whether $\maxg_i(\PotS{I}{\alpha}{v}) = + \infty$. In case of a positive answer, $\maxg_i(\PotS{I}{\alpha}{v}) $ is thus computable. If the answer is negative, as $\maxg_i(\PotS{I}{\alpha}{v}) < n$ by \emph{(ii)}, we can similarly test whether $\maxg_i(\PotS{I}{\alpha}{v}) = \const$ by considering decreasing constants $\const$ from $n-1$ to~$0$. This prove that $\barmaxg(\PotS{I}{\alpha}{v})$ is computable.
\end{proof}

\begin{proof}[Proof of Proposition~\ref{prop:regular}] 
Let us prove that each set $\PotS{I}{\alpha}{v}$ is MSO-definable by induction on $\alpha$.

For $\alpha = 0$, recall that $\PotS{I}{0}{v}$ is the set of plays starting with $v$. 
The required sentence is thus $Q_v(0) \wedge \forall x \cdot \vee_{(u,u') \in E} (Q_u(x) \wedge Q_{u'}(x+1))$.

Let $\alpha \in \IN$ be a fixed integer. By induction hypothesis, each set $\PotS{I}{\alpha}{v}$ 
is MSO-definable, and by Lemma~\ref{lem:costmaxcomput}, $\barmaxg(\PotS{I}{\alpha}{v})$ is computable. 
These sets and constants can be considered as fixed. Let us prove that $\EraseS{I}{\alpha}{v}$ is MSO-definable. 
It will follow that $\PotS{I}{\alpha +1}{v}$ is also MSO-definable. Thanks to $\barmaxg(\PotS{I}{\alpha}{v})$, 
the definition of $\rho \in \EraseS{I}{\alpha}{v}$ can be rephrased as follows: there exist $n \in \IN$, $i \in I$, 
and $u,u',v' \in V$ with $u' \neq v'$, $(u,v') \in E$, such that $\rho_n = u \in V_i$, $\rho_{n+1} = u'$ , $\chi(\rho_{n+1}) = J'$, and 
\begin{eqnarray} \label{eq:greater}
\g_i(\rho) > \maxg_i(\PotS{J'}{\alpha}{v'}) + (n+1) 
\end{eqnarray}
(see Figure~\ref{fig:decore}). Notice that (\ref{eq:greater}) implies that $i \in J'$ and $\maxg_i(\PotS{J'}{\alpha}{v'}) < + \infty$. Moreover $\maxg_i(\PotS{J'}{\alpha}{v'})$ is a fixed integer.

Let us provide an MSO-sentence $\psi$ defining $\EraseS{I}{\alpha}{v}$. The next sentence $\phi_{J',n}$ expresses that $J' \subseteq I$ is the subset of players of $I$ that did not visit their target set along $\rho_{\leq n}$:
$$
\phi_{J',n} = \left( \forall x \cdot (x \leq n) \rightarrow \neg (\vee_{j \in J'} \vee_{r \in \Target_j} Q_r(x)) \right) \wedge \left( \wedge_{j \in I \setminus J'} \exists x\leq n \cdot \vee_{r \in \Target_j} Q_r(x) \right).
$$
The next sentence $\varphi_{J',n,v',i}$ expresses that if player $i$ visits its target set along $\rho$, it is after $\maxg_i(\PotS{J'}{\alpha}{v'}) + n+1$ edges from $\rho_0$:
$$
\varphi_{J',n,v',i} = \forall x \cdot \left( \vee_{r \in \Target_i}  Q_r(x) \rightarrow \left( x > \maxg_i(\PotS{J'}{\alpha}{v'}) + n+1 \right) \right).
$$
Notice that in the previous formula, $\maxg_i(\PotS{J'}{\alpha}{v'})$ is a constant since $\PotS{J'}{\alpha}{v'}$ is a fixed set. The required formula $\psi$ is then the following one: 
$$
\exists n \cdot \bigvee_{\substack{u,u' \neq v' \in V \\ (u,v') \in E}} \bigvee_{J' \subseteq I} \bigvee_{\substack{i \in J', u \in V_i \\ \maxg_i(\PotS{J'}{\alpha}{v'}) < +\infty}} \left( Q_u(n) \wedge Q_{u'}(n+1) \wedge \phi_{J',n} \wedge \varphi_{J',n,v',i} \right).
$$
\end{proof}

By Proposition~\ref{prop:regular}, the next corollary states that one can effectively extract a lasso play from $\PotS{I}{\alpha}{v}$ that has a maximal cost for a given player $i$.

\begin{corollary} \label{cor:lasso}
For all $i \in I$, each set $\PotS{I}{\alpha_\ast}{v}$ has a computable lasso play $h \ccdot g^{\omega}$ with $\g_i(h \ccdot g^{\omega}) = \maxg_i(\PotS{I}{\alpha_\ast}{v})$. This play depends on $i$, $I$, and $v$.
\end{corollary}

\begin{proof}
Part \emph{(i)} of the proof of Lemma~\ref{lem:costmaxcomput} indicates that the set of plays $\rho \in \PotS{I}{\alpha_\ast}{v}$ with maximal cost $\g_i(\rho) = \maxg_i(\PotS{I}{\alpha_\ast}{v})$ is $\omega$-regular. Therefore, from a B\"uchi automaton accepting this set, we can extract an accepted lasso play of the form $h \ccdot g^\omega$ with the required cost. Such a play depends on $i$, $I$, and $v$ ($\alpha_*$ is fixed).
\end{proof}

\subsection{Construction of a Finite-Memory SPE} \label{subsec:finite-mem}

Thanks to the results of Sections~\ref{subsec:PotSI}-\ref{subsec:regular}, we have all the ingredients to prove that each quantitative reachability game has a computable finite-memory SPE.

\begin{proof}[Proof of Theorem~\ref{thm:main}] 
Let $(\G,v_0)$ be a quantitative reachability game. Let us summarize the results obtained previously. By Corollary~\ref{cor:existence}, each set $\PotS{I}{\alpha_\ast}{v}$ is non-empty with $v \in V$ and $I = \Sh{I}{h}$ for some history $hv\in \Hist(v_0)$, and $\PotS{\Pi}{\alpha_\ast}{v_0}$ contains all the outcomes of SPEs in $(\G,v_0)$. By Corollary~\ref{cor:integer} and Proposition~\ref{prop:regular}, we know that $\alpha_\ast \in \IN$ and each $\PotS{I}{\alpha_\ast}{v}$ is an $\omega$-regular set that can be constructed. Finally by Corollary~\ref{cor:lasso}, for all $i \in I$, one can construct a lasso play $\hg{i}{I}{v} \in \PotS{I}{\alpha_\ast}{v}$ with maximal cost $\g_i(\hg{i}{I}{v}) = \maxg_i(\PotS{I}{\alpha_\ast}{v})$.

We now show how to construct a finite-memory SPE $\bar \sigma$ from the finite set of lasso plays $\hg{i}{I}{v}$. The procedure is similar to the one developed in the proof of Theorem~\ref{thm:FolkThmWeakSPE} and more particularly of Lemma~\ref{NotEmptyToSPE}. We indicate how to adapt the proof of this lemma. Again the construction of $\bar \sigma$ is done step by step, thanks to a labeling $\gamma$ of the non-empty histories.

Initially, none of the histories is labeled. We start with history $v_0$ and with any play $\hg{i}{\Pi}{v_0} \in  \PotS{\Pi}{\alpha_\ast}{v_0}$, $i \in \Pi$. The strategy profile $\bar\sigma$ is partially defined such that $\out{\bar\sigma}_{v_0} = \hg{i}{\Pi}{v_0} $, and the non-empty prefixes $h$ of $\hg{i}{\Pi}{v_0}$ are all labeled with $\gamma(h) = (i,\Pi,v_0)$.

At the following steps, we consider a history $h'v'$ that is not yet labeled, but such that $h'$ has already been labeled. By induction, $\gamma(h') = (j,I,v)$ and there exists $hv \leq h'$ such that $\out{\SB{\bar\sigma}{h}}_{v} = \hg{j}{I}{v}$. Suppose that $\Last(h') \in V_i$ and $\Sh{I}{h'}= J'$, the proof of Lemma~\ref{NotEmptyToSPE} requires to choose\footnote{This proof states that such a play always exists.} a play $\rho' \in \PotS{J'}{\alpha_\ast}{v'}$ such that $\g_i(h'\rho') \geq \g_i(h \ccdot \out{\SB{\bar\sigma}{h}}_{v})$ (see (\ref{eq:chosen})). We simply choose $\rho' = \hg{i}{J'}{v'}$ that has maximal cost $\g_i(\hg{i}{J'}{v'}) = \maxg_i(\PotS{J'}{\alpha_\ast}{v'})$. Then we continue the construction of $\bar\sigma$ such that $\out{\SB{\bar\sigma}{h'}}_{v'} = \hg{i}{J'}{v'}$, and for all non-empty prefixes $g$ of $\hg{i}{J'}{v'}$, we define $\gamma(h'g) = (i,J',v')$.

By the proof of Lemma~\ref{NotEmptyToSPE}, the strategy profile $\bar\sigma$ is an SPE. It is finite-memory since for all $h \in \Hist_i$, $\sigma_i(h)$ only depends on $\gamma(h) = (j,I,v)$ and $\hg{j}{I}{v}$. There is a finite number of lasso plays $\hg{j}{I}{v}$, and $\gamma(h)$ (as well as $\Sh{I}{h}$) can be computed inductively as follows. Initially, $\Sh{I}{v_0}= \Pi$, and $\gamma(v_0) = (i,\Pi,v_0)$ for some chosen $i \in \Pi$. Let $h' \in \Hist_i$ and suppose that $\Sh{I}{h'}= J'$ and $\gamma(h') = (j,I,v)$. Then $\Sh{I}{h'v'} = J' \setminus\{i \mid v' \in T_i\}$. If  $h'v'$ respects
$\hg{j}{I}{v}$, i.e. $\sigma_i(h') = v'$, then $\gamma(h'v') = (j,I,v)$. Otherwise $\gamma(h'v') = (i,J',v')$ with $(i,J',v')$ computed as in the previous paragraph. 
\end{proof}

\subsection{Constrained Existence}

It remains to prove the decidability of the constrained existence of SPE for quantitative reachability games, as announced in Corollary~\ref{cor:main}. This result is easily proved on the basis of some previous properties.

\begin{proof}[Proof of Corollary~\ref{cor:main}] 
Let $(\G,v_0)$ be a game and let $\bar \const \in \IN^{|\Pi|}$ be a constant vector. In the proof of Lemma~\ref{lem:costmaxcomput}, we have seen that there exists an MSO-sentence expressing that play $\rho$ has a fixed cost $\g_i(\rho) = \const_i$. Similarly, one can express that $\g_i(\rho) \leq \const_i$ by the next sentence~$\varphi_i$: $\exists x \leq \const_i \cdot (\vee_{u \in \Target_i} Q_u(x))$. By Proposition~\ref{prop:regular}, we know that the set $\PotS{\Pi}{\alpha_\ast}{v_0}$ of outcomes of SPEs in $(\G,v_0)$ is an $\omega$-regular set, and that one can construct an MSO-sentence $\phi$ defining it. Therefore the set of outcomes of SPEs with a cost component-wise bounded by $\bar \const$ is definable by $ \wedge_{i \in \Pi} \varphi_i \wedge \phi$, and is then $\omega$-regular. Moreover, one can decide whether this set is non-empty. In case of positive answer, it contains a lasso play $h \ccdot g^\omega$. Exactly as done in Section~\ref{subsec:finite-mem}, one can construct a finite-memory SPE $\bar \sigma$ such that $\out{\bar \sigma}_{v_0} = h \ccdot g^\omega$. This concludes the proof. 
\end{proof}

\section{Games with Prefix-independent Regular Cost Functions} \label{section:prefixind}

In this section, we present a class of games for which it is decidable whether there exists a weak SPE.\footnote{Contrarily to quantitative reachability games, we do not know if a weak SPE always exists for games in this class.} The hypotheses are general conditions on the cost functions $\g_i$, $i \in \Pi$: each function $\g_i$ must be prefix-independent (see Definition~\ref{def:kindCost}), $\g_i$ has to use a finite number of values (gathered in set $\Range_i$), and the set of plays $\rho$ with a given cost $\g_i(\rho) = \const_i$ must be $\omega$-regular. 

\begin{theorem} \label{thm:existsWeakSPE}
Let $(\G,v_0)$ be an initialized game such that:
\begin{itemize}
\item each cost function $\g_i$ is prefix-independent, and with finite range $\Range_i \subset \IQ$,
\item for all $i \in \Pi$, $\const_i \in \Range_i$, and $v \in V$, the set of plays $\rho$ in $(\G,v)$ with $\g_i(\rho) = \const_i$ is an $\omega$-regular set.
\end{itemize}
\noindent
Then one can decide whether $(\G,v_0)$ has a weak SPE $\bar \sigma$ (resp. such that $\g_i(\out{\bar \sigma}_{v_0}) \leq \const_i$ forall $i$ for given $\const_i \in \Range_i$, $i \in \Pi$). In case of positive answer, one can construct such a finite-memory weak SPE.
\end{theorem}

For example, the hypotheses of this theorem are satisfied by the liminf games and the limsup games; they are also satisfied by the game of Example~\ref{ex:contrex}. We will see that the proof of this decidability result shares similar points with the proof given in the previous section for quantitative reachability games. Again, we will use the Folk Theorem for weak SPEs (see Theorem~\ref{thm:FolkThmWeakSPE}) to prove this result. The main steps of the proof are the following ones. 
\begin{itemize}
\item Given $\alpha$, the infinite number of sets $\PotPlay_\alpha(hv)$ can be replaced by the finite number of sets $\PotPlay_\alpha(v)$.
\item The fixpoint of Proposition~\ref{prop:fixpoint} is reached with some natural number $\alpha_\ast \in \IN$.
\item Each $\PotPlay_\alpha(v)$ is an $\omega$-regular set. Therefore there exists an algorithm to construct the sets $\PotPlay_{\alpha_\ast}(v)$ for all $v \in V$, and thus to decide whether they are all non-empty. For given constants $\const_i \in \Range_i$, $i \in \Pi$, one can also decide whether $\PotPlay_{\alpha_\ast}(v_0)$ has a play $\rho$ with bounded cost $\g_i(\rho) \leq \const_i$ for all $i$.
\item In case of positive answer, some lasso plays of the sets $\PotPlay_{\alpha_\ast}(v)$ allow to construct a finite-memory weak SPE (resp. with bounded cost).
\end{itemize}

To establish Theorem~\ref{thm:existsWeakSPE}, we prove a series of lemmas. The first lemma states that $\PotPlay_\alpha(hv)$ is independent of $h$. There is thus a finite number of sets  $\PotPlay_\alpha(v)$, $v \in V$, to study.

\begin{lemma} \label{lem:prefixind}
$\PotPlay_\alpha(hv) = \PotPlay_\alpha(v)$ for all $hv \in \Hist(v_0)$.
\end{lemma}

\begin{proof}
The proof can be easily done by induction on $\alpha$. It uses the definition of $\ErasePlay_\alpha(hv)$ and the hypothesis of Theorem~\ref{thm:existsWeakSPE} that each cost function $\g_i$ is prefix-independent. 
\end{proof}

As each cost function $\g_i$ is supposed to have finite range in Theorem~\ref{thm:existsWeakSPE}, we can give the next definition that indicates the maximum costs for plays in $\PotPlay_{\alpha}(v)$ (resp. starting with $vv'$, for some given $(v,v') \in E$). Recall that a similar definition was given in case of quantitative reachability games (see Definition~\ref{def:costmax} and~(\ref{eq:adapt})).

\begin{definition} \label{def:costmaxbis}
Given $\PotPlay_\alpha(v)$ and $(v,v') \in E$, we define for each $i \in \Pi$:
\begin{itemize}
\item $\maxg_i(\PotPlay_\alpha(v)) =  \max \{ \g_i(\rho) \mid \rho \in \PotPlay_\alpha(v) \}.$
\item $\maxg_i(\PotPlay_\alpha(v),v') = \max \{ \g_i(\rho) \mid \rho \in \PotPlay_\alpha(v) \mbox{ and } \rho_0\rho_1 = vv' \}.$
\end{itemize}
\end{definition}

\noindent 
In this definition, the $\max$ is equal to $-\infty$ if it applies to an empty set.

The sequence $(\barmaxg(\PotPlay_\alpha(v),v'))_\alpha$ is nonincreasing for the component-wise ordering over $(\IQ \cup \{-\infty\})^{|\Pi|}$.\footnote{More precisely each component in $(\IQ \cup \{-\infty\})^{|\Pi|}$ is restricted to $\Range_i \cup \{-\infty\}$.} Therefore it reaches a fixpoint as for $(\PotPlay_\alpha(v))_\alpha$. The following lemma relates these sequences.

\begin{lemma} \label{lem:lienbis}
\begin{itemize}
\item If $\PotPlay_\alpha(v) = \PotPlay_{\alpha +1}(v)$, then for all $(v,v') \in E$, $\barmaxg(\PotPlay_\alpha(v),v') = \barmaxg(\PotPlay_{\alpha + 1}(v),v')$.
\item If $\PotPlay_\alpha(v) \neq \PotPlay_{\alpha + 1}(v)$, then there exists $(u,u') \in E$ such that $\barmaxg(\PotPlay_\alpha(u),u') \neq \barmaxg(\PotPlay_{\alpha + 1}(u),u')$.
\end{itemize} 
\end{lemma}

\begin{proof}
The proof is similar to the proof of Lemma~\ref{lem:lien}. Without mentioning it, we will repeatedly use Lemma~\ref{lem:prefixind} and the hypothesis of Theorem~\ref{thm:existsWeakSPE} that the cost functions are prefix-independent.

The first statement is immediate from definition of $\barmaxg$. For the second statement, consider $\rho  = \rho_0\rho_1 \ldots \in \ErasePlay_\alpha(v)$. Then by definition of $\ErasePlay_\alpha(v)$, there exist $n \in \IN$, $i \in \Pi$, 
and $u,u',v' \in V$ with $u' \neq v'$, $(u,v') \in E$, such that $\rho_n = u \in V_i$, $\rho_{n+1} = u'$, and
$$\forall \rho' \in \PotPlay_{\alpha}(v'): \g_i(\rho) > \g_i(\rho').$$ 
(See Figure~\ref{fig:decore} adapted to the context of Lemma~\ref{lem:lienbis}). Let us prove that $\barmaxg(\PotPlay_{\alpha}(u),u') \neq \barmaxg(\PotPlay_{\alpha + 1}(u),u')$. As $\rho \in \PotPlay_{\alpha}(v)$, then $\rho_{\geq n} \in \PotPlay_{\alpha}(u)$ by Lemma~\ref{lem:utile}, which implies that $\maxg_i(\PotPlay_{\alpha}(u),u') \geq \g_i(\rho_{\geq n})$. 
Let $\varrho \in \PotPlay_{\alpha}(u)$ be such that $\varrho$ starts with edge $(u,u')$ and has maximal cost $\maxg_i(\PotPlay_{\alpha}(u),u')$. One gets 
$$\g_i(\varrho) = \maxg_i(\PotPlay_{\alpha}(u),u') \geq \g_i(\rho_{\geq n}) = \g_i(\rho) > \g_i(\rho').$$ 
Hence, using the same set $\PotPlay_{\alpha}(v')$ as for $\rho$, it follows that $\varrho \in \ErasePlay_{\alpha}(u)$ for all such plays $\varrho$. Therefore $\PotPlay_{\alpha +1}(u) \subsetneq \PotPlay_{\alpha}(u)$ and $\maxg_i(\PotPlay_{\alpha}(u),u') < \maxg_i(\PotPlay_{\alpha +1}(u),u')$.
\end{proof}

As a consequence, the ordinal $\alpha_\ast$ of Proposition~\ref{prop:fixpoint} is an integer. The proof is the same as for Corollary~\ref{cor:integer}.

\begin{corollary} \label{cor:integerbis}
There exists an integer $\alpha_\ast$ such that $\PotPlay_{\alpha_\ast}(v) = \PotPlay_{\alpha_\ast +1}(v)$ for all $v \in V$.
\end{corollary}

As done for quantitative reachability games, let us now prove that the sets $\PotPlay_\alpha(v)$ are $\omega$-regular for all $\alpha$ and $v$. 

\begin{lemma} \label{lem:regularbis}
Each $\PotPlay_\alpha(v)$ is an $\omega$-regular set.
\end{lemma}

\begin{proof}
The proof is similar to the proof of Lemma~\ref{lem:costmaxcomput} and Proposition~\ref{prop:regular} (It is even simpler). Recall that as soon as $\PotPlay_\alpha(v)$ is empty, then $\PotPlay_\beta(v) = \emptyset$ for all $\beta \geq \alpha$.

Like in Lemma~\ref{lem:costmaxcomput}, we first prove that if $\PotPlay_{\alpha}(v)$ is MSO-definable, then $\maxg_i(\PotPlay_{\alpha}(v))$ is computable for each $i \in \Pi$. Let $\phi$ be an MSO-sentence defining $\PotPlay_{\alpha}(v)$. One can decide whether $\PotPlay_{\alpha}(v)$ is empty. If this is the case, then $\maxg_i(\PotPlay_{\alpha}(v)) = -\infty$ for all $i$. Suppose that $\PotPlay_{\alpha}(v)\neq \emptyset$, and let $i \in \Pi$ and $\const \in \Range_i$. By hypothesis, the set of plays $\rho$ in $(\G,v)$ with cost $\g_i(\rho) = \const$ is $\omega$-regular and thus MSO-definable by a sentence $\varphi_{\const,i}$. We can thus decide whether $\PotPlay_{\alpha}(v)$ has a play $\rho$ with cost $\g_i(\rho) = \const$, thanks to sentence $\phi \wedge \varphi_{\const,i}$. Therefore, by considering decreasing constants $\const \in \Range_i$, we can decide whether $\maxg_i(\PotPlay_{\alpha}(v)) = \const$. This shows that $\maxg_i(\PotPlay_{\alpha}(v))$ is computable.

Let us now prove that each set $\PotPlay_\alpha(v)$ is MSO-definable by induction on~$\alpha$. For $\alpha = 0$, we use the same defining MSO-sentence as in the proof of Proposition~\ref{prop:regular}: 
$$Q_v(0) \wedge \forall x \cdot \vee_{(u,u') \in E} (Q_u(x) \wedge Q_{u'}(x+1)).$$
Let $\alpha \in \IN$ be a fixed integer. By induction hypothesis, each set $\PotPlay_\alpha(v)$ 
is MSO-definable, and $\maxg_i(\PotS{I}{\alpha}{v})$, $i \in \Pi$, is computable by the first part of the proof. 
These sets and constants can be considered as fixed. The only case to consider is $\PotPlay_\alpha(v) \neq \emptyset$ (recall that this property is decidable). To show that $\PotPlay_{\alpha +1}(v)$ is also MSO-definable, it is enough to prove that $\ErasePlay_{\alpha}(v)$ is MSO-definable. Recall that $\rho \in \ErasePlay_{\alpha}(v)$ iff there exist $n \in \IN$, $i \in \Pi$, and $u,u',v' \in V$ with $u' \neq v'$, $(u,v') \in E$, such that $\rho_n = u \in V_i$, $\rho_{n+1} = u'$, and $\forall \rho' \in \PotPlay_{\alpha}(v')$: $\g_i(\rho) > \g_i(\rho')$. 
The last condition can be replaced by $\g_i(\rho) > \maxg_i(\PotPlay_{\alpha}(v')) \neq -\infty$.\footnote{Set $\PotPlay_{\alpha}(v')$ must be non-empty.}
Let us provide an MSO-sentence $\psi$ defining $\ErasePlay_{\alpha}(v)$: 
$$
\exists n \cdot \bigvee_{\substack{i \in \Pi, u\in V_i \\ u' \neq v' \in V \\ (u,v') \in E}} \quad \bigvee_{\substack{\const \in \Range_i \\ \const > \maxg_i(\PotPlay_{\alpha}(v')) \neq -\infty}}
\left( Q_u(n) \wedge Q_{u'}(n+1) \wedge \varphi_{\const,i} \right).
$$
\end{proof}

We get the next corollary. The proof is the same as for Corollary~\ref{cor:lasso}.

\begin{corollary} \label{cor:lassobis}
If $\PotPlay_\alpha(v) \neq \emptyset$, then one can compute a lasso play $h \ccdot g^{\omega}$ in $\PotPlay_\alpha(v)$ with $\g_i(h \ccdot g^{\omega}) = \maxg_i(\PotPlay_\alpha(v))$. This play depends on $i$ and $v$.
\end{corollary}

We are now able to prove the main result of this section.

\begin{proof}[Proof of Theorem~\ref{thm:existsWeakSPE}]
Let $(\G,v_0)$ be a game satisfying the hypotheses of Theorem~\ref{thm:existsWeakSPE}. Let us summarize the results of the previous lemmas. We know that $\alpha_\ast \in \IN$ and that one can construct the sets $\PotPlay_{\alpha_\ast}(v)$, $v \in V$. As these sets are $\omega$-regular, one can decide whether they are all non-empty. In case of positive answer, there exists a weak SPE in $(\G,v_0)$ by Theorem~\ref{thm:FolkThmWeakSPE}. If in addition some constants $\const_i \in \Range_i$ are given, then the set $\PotPlay_{\alpha_\ast}(v_0) \cap \{\rho \mbox{ in } (\G,v_0) \mid \g_i(\rho) \leq \const_i, \forall i \in \Pi \}$ is also $\omega$-regular. Hence one can also decide whether this set is non-empty and thus whether there exists a weak SPE in $(\G,v_0)$ with cost component-wise bounded by $\bar \const$. This establishes the first part of Theorem~\ref{thm:existsWeakSPE}. 

Suppose that such a weak SPE exists, then let us show that we can construct a weak SPE that is finite-memory with the same construction as in the proof of Theorem~\ref{thm:main}. By Corollary~\ref{cor:lassobis}, for all $i \in \Pi$, $v \in V$, one can construct a lasso play $\hgbis{i}{v} \in \PotPlay_{\alpha_\ast}(v)$ with maximal cost $\g_i(\hgbis{i}{v}) = \maxg_i(\PotPlay_{\alpha_\ast}(v))$. The construction of a finite-memory SPE $\bar \sigma$ from the finite set of lasso plays $\hgbis{i}{v}$ is conducted as in the proof of Lemma~\ref{NotEmptyToSPE}. It is done step by step thanks to a labeling $\gamma$ of the non-empty histories. 

Initially, none of the histories is labeled. We start with history $v_0$ and with any play $\hgbis{i}{v_0} \in  \PotPlay_{\alpha_\ast}(v_0)$, $i \in \Pi$.\footnote{When some constants $\const_i \in \Range_i$ are additionally given, play $\hgbis{i}{v_0}$ must be replaced by any lasso play in $\PotPlay_{\alpha_\ast}(v_0) \cap \{\rho \mbox{ in } (\G,v_0) \mid \g_i(\rho) \leq \const_i, \forall i \in \Pi \}$.} The strategy profile $\bar\sigma$ is partially defined such that $\out{\bar\sigma}_{v_0} = \hgbis{i}{v_0}$, and the non-empty prefixes $h$ of $\hgbis{i}{v_0}$ are all labeled with $\gamma(h) = (i,v_0)$. 

At the following steps, we consider a history $h'v'$ that is not yet labeled, but such that $h'$ has already been labeled. By induction, $\gamma(h') = (j,v)$ and there exists $hv \leq h'$ such that $\out{\SB{\bar\sigma}{h}}_{v} = \hgbis{j}{v}$. Suppose that $\Last(h') \in V_i$, the proof of Lemma~\ref{NotEmptyToSPE} requires to choose a play $\rho' \in \PotPlay_{\alpha_\ast}(v')$ such that $\g_i(h'\rho')  = \g_i(\rho') \geq \g_i(h \ccdot \out{\SB{\bar\sigma}{h}}_{v})$. 
We choose $\rho' = \hgbis{i}{v'}$ with maximal cost $\g_i(\rho') = \maxg_i(\PotPlay_{\alpha_\ast}(v'))$. Then we continue the construction of $\bar\sigma$ such that $\out{\SB{\bar\sigma}{h'}}_{v'} = \hgbis{i}{v'}$, and for all non-empty prefixes $g$ of $\hgbis{i}{v'}$, we define $\gamma(h'g) = (i,v')$.

We know by Lemma~\ref{NotEmptyToSPE} that $\bar\sigma$ is an SPE. It is finite-memory because it only depends on the finite number of lasso plays $\hgbis{j}{v}$, and the labeling $\gamma$ that can be computed inductively as follows. Initially, $\gamma(v_0) = (i,v_0)$ for some chosen $i \in \Pi$. Let $h' \in \Hist_i$ and suppose that $\gamma(h') = (j,v)$. If $h'v'$ respects
$\hgbis{j}{v}$, then $\gamma(h'v') = (j,v)$, otherwise $\gamma(h'v') = (i,v')$ (as explained in the previous paragraph). 
This establishes the second part of Theorem~\ref{thm:existsWeakSPE}.
\end{proof}

\section{Conclusion and Future Work}

In this article, we have studied the existence of (weak) SPEs in quantitative games. We have proposed a Folk Theorem for weak SPEs, and a weaker version for SPEs.
To illustrate the potential of this theorem, we have given two applications. The first one is concerned with quantitative reachability games for which we have provided an algorithm to compute a finite-memory SPE, and a second algorithm for deciding the constrained existence of a (finite-memory) SPE. The second application is concerned with another large class of games for which we have proved that the (constrained) existence of a (finite-memory) weak SPE is decidable.

Future possible directions of research are the following ones. We would like to study the complexities of the problems studied for the two classes of games. We also want to investigate the application of our Folk Theorem to other classes of games. The example of Figure~\ref{fig:gameNoSPE} is a game with a weak SPE but no SPE (see Example~\ref{ex:contrex}). Recall that for this game, the cost $\g_i(\rho)$ can be seen as either the mean-payoff, or the liminf, or the limsup, of the weights of $\rho$. We do not know if games with this kind of payoff functions always have a weak SPE or not.

\nocite{Simpson}


\bibliographystyle{plain}
\bibliography{weakSPE}
\end{document}